\documentclass[10pt]{amsart}
\usepackage{amsmath}
\usepackage{latexsym}
\usepackage{amsfonts}
\usepackage{amssymb}
\usepackage{color}
\usepackage{bbm,dsfont}

\newtheorem{proposition}{Proposition}
\newtheorem{theorem}{Theorem}
\newtheorem{lemma}{Lemma}
\newtheorem{corollary}{Corollary}

\theoremstyle{definition}

\newtheorem{definition}{Definition}

\newcommand{\frecc}{\rightarrow} 

\newcommand{\de}{\,{\rm d}} 

\newcommand{\lft}{\left(} 
\newcommand{\rgt}{\right)} 
\newcommand{\lfq}{\left[} 
\newcommand{\rgq}{\right]}

\newcommand{\R}{\mathbb R} %real
\newcommand{\T}{\mathbb{T}} %torus
\newcommand{\C}{{\mathbb C}} %complex
\newcommand{\Z}{\mathbb{Z}} %int
\newcommand{\half}{\frac{1}{2}}
\newcommand{\ff}{\mathcal{F}} %Fourier transform
\newcommand{\hG}{\widehat{G}} %dual group of G
\newcommand{\hY}{Y} %Borel set in the dual of G
\newcommand{\hlam}{\hat{\lambda}} %dual measure of lambda
\newcommand{\WH}{\mathfrak{H}_G} %Weyl-Heisenberg group

\newcommand{\hh}{\mathcal{H}} %Hilbert space H 
\newcommand{\hi}{\hh} %Hilbert space H
\newcommand{\hhat}{\hat{\hh}} %Hilbert space H^ 
\newcommand{\kk}{\mathcal{K}} %Hilbert space K
\newcommand{\vv}{\mathcal{V}} %Hilbert space V
\newcommand{\id}{\mathbbm{1}} %identity operator

\newcommand{\sh}{\mathcal{S(H)}} %states
\newcommand{\trh}{\mathcal{T}(\mathcal{H})} %trace class
\newcommand{\trhh}{\mathcal{T}(\hh\otimes\hh)} %trace class on H tensor H
\newcommand{\lh}{\mathcal{L(H)}} %bounded operators
\newcommand{\ti}[1]{\mathcal{T} \left( #1 \right)} %trace class in generic Hilbert space
\newcommand{\lam}{\lambda}
\newcommand{\ip}[2]{\left\langle\,#1\,|\,#2\,\right\rangle} %inner product LINEAR W.R.T 2nd ARGUMENT
\newcommand{\scal}[2]{\left\langle\,#2\,|\,#1\,\right\rangle} %scalar product USED ONLY IN THE PROOF SECTION
\newcommand{\ketbra}[2]{|#1\rangle\langle #2|} %ketbra
\newcommand{\kb}{\ketbra} %ketbra
\newcommand{\no}[1]{\left\|#1\right\|} %norm
\newcommand{\tr}[1]{\mathrm{tr}\left[ #1 \right]} %trace
\newcommand{\trped}[2]{{\rm tr}_{#1} \left[ #2 \right]} %partial trace in #1
\newcommand{\Ao}{\mathsf{A}} %povm
\newcommand{\Bo}{\mathsf{B}} %povm
\newcommand{\Co}{\mathsf{C}} %povm
\newcommand{\Ro}{\mathsf{R}} %pvm
\newcommand{\Eo}{\mathsf{E}} %povm
\newcommand{\Fo}{\mathsf{F}} %povm
\newcommand{\Mo}{\mathsf{M}} %vector measure
\newcommand{\So}{\mathsf{S}} %spin
\newcommand{\ii}{\mathcal{I}} %instrument
\newcommand{\jj}{\mathcal{J}} %instrument
\newcommand{\mm}{\mathcal{M}} %measure space

\newcommand{\bor}[1]{\mathcal{B}(#1)} %borel space
\newcommand{\borel}[1]{{\mathcal B}\lft \real\rgt}

\newcommand{\ldue}[1]{L^2(#1)} %L two
\newcommand{\luno}[1]{L^1(#1)} %L one

\newcommand{\va}{\mathbf{a}} %a
\newcommand{\vb}{\mathbf{b}} %b
\newcommand{\vr}{\mathbf{r}} %r
\newcommand{\vsigma}{\boldsymbol{\sigma}} %sigma
\begin{document}

\setlength\arraycolsep{2pt}

\title[Sequential measurements]{Sequential measurements of conjugate observables}

\author[Carmeli]{Claudio Carmeli}
\address{Claudio Carmeli, Dipartimento di Fisica, Universit\`a di Genova, and Istituto Nazionale di Fisica Nucleare, Sezione di Genova, Via Dodecaneso 33, 16146 Genova, Italy}
\email{claudio.carmeli@gmail.com}

\author[Heinosaari]{Teiko Heinosaari}
\address{Teiko Heinosaari, Turku Centre for Quantum Physics, Department of Physics and Astronomy, University of Turku, 20014 Turku, Finland}
\email{teiko.heinosaari@utu.fi}

\author[Toigo]{Alessandro Toigo}
\address{Alessandro Toigo, Dipartimento di Matematica ``Francesco Brioschi'', Politecnico di Milano, Piazza Leonardo da Vinci 32, 20133 Milano, and Istituto Nazionale di Fisica Nucleare, Sezione di Milano, Via Celoria 16, 20133 Milano, Italy}
\email{alessandro.toigo@polimi.it}

\maketitle{}

\begin{abstract}
We present a unified treatment of sequential measurements of two conjugate observables.
Our approach is to derive a mathematical structure theorem for all the relevant covariant instruments.
As a consequence of this result, we show that every Weyl-Heisenberg covariant observable can be implemented as a sequential measurement of two conjugate observables.
This method is applicable both in finite and infinite dimensional Hilbert spaces, therefore covering sequential spin component measurements as well as position-momentum sequential measurements.
\end{abstract}

%%%%%%%%%%%%%%%%%%%%%%
\section{Introduction}\label{sec:intro}
%%%%%%%%%%%%%%%%%%%%%%

A sharp measurement of position affects the state of a quantum system in a dramatic way; any subsequent measurement can give only redundant information on the initial state of the system. 
On the other hand, if we perform a measurement which does not disturb the system at all, then nothing on the initial state can be inferred from the measurement outcome statistics. 
This latter consequence of quantum theory is usually referred as ``no information without disturbance''.

Obviously, there is no need to restrict to these two extremes and actually the intermediate cases are more interesting and practical.  
In particular, in a sequential measurement the aim is to perform several measurements in succesion and gather additional information on the initial state at each step.
The benefit of sequential measurements is that we also get correlations and not just separate measurement outcome distributions. 
For this reason a sequential measurement scheme can give more information on the initial state than separate measurements together.

Peforming measurements sequentially can be seen as a method to combine some simple measurements in order to realize a more complicated measurement.
Especially in quantum information theory, it has become evident that sharp observables (described by projection valued measures) are not enough for all purposes \cite{QCQI00}. 
One needs also more involved observables, generally described by positive operator valued measures \cite{OQP97,QTOS76,PSAQT82}, in order to perform various tasks.
This opens up the question how to realize these more complicated measurements.

One particular class of interesting observables consist of \emph{covariant phase space observables}.
In our context, covariance refers to a certain type of symmetry property with respect to the (finite or infinite) Weyl-Heisenberg group.
The most prominent example of a covariant phase space observable is the Q-function, which was first introduced by Husimi for infinite dimensional systems \cite{Husimi40} and later studied also in the connection of finite dimensional systems \cite{OpBuBaDr95}.

It is known that some of the covariant phase space observables are \emph{informationally complete} \cite{AlPr77a}, meaning that the obtained measurement outcome distribution determines the initial state completely. 
(Conditions for informational completeness have been discussed e.g. in \cite{CaDeLaLe00} for infinite dimensional and in \cite{DaPeSa04} for finite dimensional cases, respectively.)
Informational completeness is the main reason for making the class of covariant phase space observables so interesting.
An additional impetus is coming from the special role of the Weyl-Heisenberg group in the topic of SIC-POVMs \cite{ReBlScCa04}.
In fact, most of the known SIC-POVMs are covariant phase space observables \cite{Appleby05}, \cite{Flammia06}. 

Additional physical insight can be gained when we recognize the connection of covariant phase space observables to joint measurements of conjugate observables.
It is known that a covariant phase space observable in infinite dimension gives marginals which are approximate position and momentum observables \cite{QTOS76},\cite{Holevo73} and it can therefore be interpreted as a joint measurement of \emph{unsharp} position and momentum observables \cite{Busch85}.
This kind of joint measurement is limited, of course, by the Heisenberg's uncertainty relation \cite{BuHeLa07}.
It is also known that a suitable sequential scheme of an unsharp position measurement followed by a (sharp) momentum measurement yields a covariant phase space observable \cite{OQP97}. 
 
In this work we give a systematic treatment of sequential measurements of conjugated pairs and we demonstrate that in this way one can implement \emph{all} covariant phase space observables.
Our investigation is formulated by starting from a (locally compact and second countable) abelian group $G$ and its dual group $\hG$, then passing to the associated Weyl-Heisenberg group $\WH$.
The results are therefore general and applicable, in particular, to the common situations of the position-momentum pair and mutually unbiased bases. 

\emph{Outline.} The necessary basic concepts are shortly reviewed in Section \ref{sec:sequential}.
In Section \ref{sec:receipt} we explain a scheme how an instrument with appropriate covariance properties leads to an implementation of a covariant phase space observable.
In Section \ref{sec:structure} we derive a general structure theorem for these covariant instruments, and as a corollary this leads to the conclusion that all covariant phase space observables can be realized with this sequential method.
Finally, in Section \ref{sec:examples} we illustrate the results in two concrete cases of position-momentum and orthogonal spin components.
Section \ref{app:proof:th:gen} contains the proofs for the most technical parts of the paper.

\emph{Notations.} In the following $\hi$ is a Hilbert space (and we always assume that our Hilbert spaces are separable).
We denote by $\lh$ and $\trh$ the Banach spaces of bounded operators and trace class operators on $\hi$, respectively. 
We let $\no{\cdot}_\infty$ be the uniform norm in $\lh$, and $\no{\cdot}_1$ be the trace class norm in $\trh$. 
The cone of positive elements in $\trh$ is denoted by $\trh_+$, and $\sh$ is the set of states on $\hh$, i.~e., the convex closed subset of trace $1$ elements in $\trh_+$.

%%%%%%%%%%%%%%%%%%%%%%%%%%%%%%%%%%%%%%
\section{Instruments and sequential measurements}\label{sec:sequential}
%%%%%%%%%%%%%%%%%%%%%%%%%%%%%%%%%%%%%%

The mathematical framework for sequential measurements was introduced by Davies and Lewis in \cite{DaLe70}.
In the following we briefly summarize the essential concepts.

A linear mapping $\Phi:\trh\to\trh$ is an \emph{operation} if it is completely positive and satisfies $0\leq \tr{\Phi(\varrho)}\leq 1$ for all $\varrho\in\sh$.
An operation $\Phi$ describes a conditional state change in the following way: if the initial state is $\varrho$, the final (unnormalized) state is $\Phi(\varrho)$, provided this is a nonzero operator. 
The number $\tr{\Phi(\varrho)}$ is the probability for the occurrence of the particular event associated with $\Phi$.
An operation $\Phi$ defines its dual mapping $\Phi^\ast$ via the formula
\begin{equation*}
\tr{\varrho\Phi^\ast(A)} = \tr{\Phi(\varrho)A} \qquad \forall \varrho\in\sh,A\in\lh \, .
\end{equation*}
The dual operation $\Phi^\ast$ acts on $\lh$ and it describes the same event than $\Phi$ but in the Heisenberg picture.

In a measurement process, the relevant events are of the type `the measurement gave an outcome belonging to a set $X$'. 
To describe all the corresponding conditional state changes, let $\Omega$ be the set consisting of all measurement outcomes. 
The set $\Omega$ is assumed to be a locally compact topological space, which is Hausdorff and satisfies the second axiom of countability (lcsc space, in short). 
The Borel $\sigma$-algebra of $\Omega$ is denoted by $\bor{\Omega}$, and the Borel sets are identified with the possible events in the measurement.

\begin{definition}\label{def:ins}
An \emph{instrument} on $\Omega$ is a mapping $\ii : X \mapsto \ii_X$ from $\bor{\Omega}$ to the set of operations on $\trh$ such that for each state  $\varrho\in\sh$, the mapping $X\mapsto\tr{\ii_X(\varrho)}$ is a probability measure.
\end{definition}

An instrument gives a description of two things: measurement outcome probabilities and conditional state changes. 
It is the general form of a channel yielding classical information together with a quantum output.
On the other hand, it can be shown that every instrument arises from a measurement process \cite{Ozawa84}.

After a measurement has been performed, we can directly (i.~e.~without further measurements) observe only the measurement outcome distribution. 
An instrument $\ii$ on $\Omega$ determines a unique \emph{associated observable} $\Eo^{\ii}$, given by
\begin{equation}
\Eo^{\ii}(X)=\ii^\ast_X(\id) \qquad \forall X\in\bor{\Omega} \, .
\end{equation}
If the system is initially in a state $\varrho$, then the measurement gives an outcome from a set $X$ with probability
\begin{equation}
\tr{\varrho\Eo^{\ii}(X)}=\tr{\varrho\ii^{\ast}_X(\id)}=\tr{\ii_X(\varrho)} \, .
\end{equation}
Mathematically speaking, the mapping $X\mapsto\Eo^{\ii}(X)$ is a \emph{positive operator valued measure} (POVM). 
We recall the following standard definition \cite{OQP97,QTOS76,PSAQT82}.
\begin{definition}
An {\em observable} on $\Omega$ is a mapping $\Eo : \bor{\Omega} \to \lh$ such that for all states $\varrho\in\sh$ the mapping $X\mapsto \tr{\varrho\Eo (X)}$ is a probability measure. 
The observable $\Eo$ is {\em sharp} if $\Eo(X\cap Y) = \Eo(X) \Eo(Y)$ for all $X,Y$ (i.~e., if and only if $\Eo$ is a projection valued measure).
\end{definition}

Suppose that two measurements, described by instruments $\ii$ and $\ii'$, are performed sequentially.
This means that the second measurement is performed on the perturbed state.
The operation corresponding to the event 'the first measurement led to a measurement outcome from a set $X$ and the second measurement led to a measurement outcome from a set $Y$' is the composite mapping $\ii'_Y\circ \ii_X$.
As shown in \cite{DaLe70}, there exists a unique instrument $\jj$ on $\bor{\Omega\times\Omega'}$ such that $\jj_{X\times Y}=\ii'_Y\circ \ii_X$ for all $X\in\bor{\Omega}$ and $Y\in\bor{\Omega'}$; this is called the {\em composition of $\ii'$ following $\ii$} and it gives the mathematical description of the sequential measurement. 

The observable $\Eo^{\jj}$ on $\Omega\times\Omega'$, associated to the composition instrument $\jj$, satisfies
\begin{equation}\label{eq:joint-obs}
\Eo^{\jj} (X\times Y) = \jj_{X\times Y}^\ast(\id) = \ii_X^\ast \left[ \Eo^{\ii'} (Y) \right] 
\end{equation}
for all $X\in\bor{\Omega}$ and  $Y\in\bor{\Omega'}$.
We observe that $\Eo^{\jj}$ depends on $\ii'$ only through the associated observable $\Eo^{\ii'}$. 
The marginal observables are 
\begin{equation*}
\Eo^{\jj}(X \times \Omega') = \ii_X^\ast (\id) = \Eo^{\ii}(X) \quad \forall X\in\bor{\Omega}
\end{equation*}
and
\begin{equation*}
\Eo^{\jj}(\Omega \times Y) = \ii_\Omega^\ast \left[ \Eo^{\ii'}(Y) \right] \quad \forall Y\in\bor{\Omega'} \, .
\end{equation*}
The fact that the first measurement disturbs the second measurement is manifested in the difference of the operators $\ii_\Omega^\ast \left[ \Eo^{\ii'}(Y) \right]$ and  $\Eo^{\ii'}(Y)$.
If the Hilbert space $\hi$ is finite dimensional and the set of outcomes $\Omega'$ is finite, then one can even quantify the least amount of disturbance induced on $\Eo^{\ii'}$ by an $\Eo^{\ii}$-measurement \cite{HeWo10}. 

Clearly, a sequential measurement scheme is motivated only if we can gain additional information on the initial state at each step.
So, after the first measurement is performed, is it possible to learn something more on the initial state by performing another measurement?
In particular, does the joint observable $\Eo^{\jj}$ defined in eq.~\eqref{eq:joint-obs} give more information than the first observable $\Eo^{\ii}$ alone?
The answer evidently depends on the way the first measurement was carried out, i.~e., it depends on the structure of the instrument $\ii$.
However, some observables allow \emph{only} instruments that trivialize all subsequent measurements.
For instance, if an observable $\Eo$ consists of countable number of rank-1 operators, then any measurement of $\Eo$ disturbs the initial state of the system in a way that all subsequent measurements can only give redundant information \cite{HeWo10}.
As another example, suppose that each $\Eo(X)$ is a projection and that the set $\{\Eo(X) \mid X \in\bor{\Omega} \}$ generates a maximal abelian von Neumann subalgebra of $\lh$.  
It then follows that any observable $\Fo$ that satisfies the marginal condition $\Fo(X\times \Omega')=\Eo(X)$ for all $X\in\bor{\Omega}$ can actually be obtained from $\Eo$ by smearing it \cite{JePu07}.
In physical terms these examples mean that in order to have a useful sequential measurement, one has to allow additional imprecision in the first measurement to make it less violent.

%%%%%%%%%%%%%%%%%%%%%%%%%%%%%%%%%%%%%%%%%%%%%
\section{Sequential measurement of conjugate observables}\label{sec:receipt}
%%%%%%%%%%%%%%%%%%%%%%%%%%%%%%%%%%%%%%%%%%%%

In this section we first explain a general setup (Subsec. \ref{sec:setup}) and then specify it in the case of conjugate observables (Subsec. \ref{sec:receipt-conjugate}).
As it will be then explained, the Weyl-Heisenberg group is the most convenient tool in our investigation (Subsec. \ref{sec:weyl}). 

%%%%%%%%%%%%%%%%%%%%%%%%%%%%%%%%%%%%%%%%
\subsection{General setup}\label{sec:setup}
%%%%%%%%%%%%%%%%%%%%%%%%%%%%%%%%%%%%%%%%

Suppose we try to study the system by performing a sequential measurement of two observables. 
It seems reasonable to choose some obsevables which are, in some sense, very different from each other, as then they possibly lead to a complete picture of the state of the system. 
We are thus motivated to consider a pair of observables with `reciprocal' covariance properties.
We will specify the structure of canonically conjugated observables in Subsec. \ref{sec:receipt-conjugate}, but we first explain the general setup which does not depend on the detailed structure but only on the interplay between various symmetry conditions.

We assume that an observable $\Ao$ is based on a lcsc group $G$ and another observable $\Bo$ is based on a lcsc group $H$.
We further assume that there are unitary representations $U$ of $G$ and $V$ of $H$ such that
\begin{equation}\label{eq:A-covar}
U_g \Ao(X) U^\ast_g = \Ao(gX) \, , \quad V_{h} \Ao(X) V^\ast_{h} = \Ao(X) 
\end{equation} 
and
\begin{equation}\label{eq:B-covar}
V_{h} \Bo(Y) V^\ast_{h} = \Bo(hY) \, , \quad U_g \Bo(Y) U^\ast_g = \Bo(Y)
\end{equation}
for all $X\in \bor{G}$, $Y\in \bor{H}$, $g\in G$, $h\in H$.
We have denoted $gX=\{ gx \mid x \in X\}$ and $hY=\{ hy \mid y \in Y\}$.
The conditions \eqref{eq:A-covar}-\eqref{eq:B-covar} mean that $\Ao$ is $U$-covariant and $V$-invariant, while $\Bo$ is $V$-covariant and $U$-invariant.

We would like to have a sequential measurement scheme which provides the measurement outcome distributions of $\Ao$ and $\Bo$.
However, depending on $\Ao$ and $\Bo$, their exact sequential realization may be impossible (see the end of Sec.\ref{sec:sequential}).
For this reason, we will first only concentrate on the symmetry properties \eqref{eq:A-covar}-\eqref{eq:B-covar} and hope that this gives, if not $\Ao$ and $\Bo$, at least something quite similar.

As a mathematical problem, our task is to find a suitable instrument describing this measurement scheme.
Suppose there exists an instrument $\ii$ based on $G$ and satisfying
\begin{equation}\label{eq:ins-covar}
\ii_{gX}(\varrho) = U_g \ii_X ( U^\ast_g \varrho U_g ) U^\ast_g
\end{equation}
and
\begin{equation}\label{eq:ins-invar}
\ii_{X}(\varrho) = V_{h} \ii_X ( V^\ast_{h} \varrho V_{h} )   V^\ast_h
\end{equation}
for all $X\in \bor{G}$, $g\in G$, $h\in H$ and $\varrho\in \sh$. 
These properties are motivated by the $U$-covariance and $V$-invariance of $\Ao$.
The conditions \eqref{eq:ins-covar}-\eqref{eq:ins-invar} can be merged into the following single condition
\begin{equation}\label{eq:ins-UV}
\ii_{gX}(\varrho) = U_g V_h \ii_X (V^\ast_h U^\ast_g \varrho U_g V_h) V^\ast_h U^\ast_g \, .
\end{equation}
We then consider a sequential measurement in which the first measurement is described by $\ii$ and it is followed by any measurement of $\Bo$. 
Hence, the observable $\Co$ describing the sequential measurement satisfies
\begin{equation}\label{eq:G-seq}
\Co(X\times Y) = \ii^\ast_X [\Bo(Y)] \qquad X\in\bor{G}, Y\in\bor{H} \, .
\end{equation}  
We denote the marginals of $\Co$ by $\widetilde{\Ao}$ and $\widetilde{\Bo}$, i.~e.,
\begin{eqnarray}
\widetilde{\Ao}(X) &=& \Co(X\times H) = \ii^\ast_X(\id) \label{eq:marg-A}\\
\widetilde{\Bo}(Y) &=& \Co(G\times Y) = \ii^\ast_G[\Bo(Y)] \, . \label{eq:marg-B}
\end{eqnarray}
Generally, $\widetilde{\Ao} \neq \Ao$ (since we have not required $\Ao$ is associated to $\ii$) and $\widetilde{\Bo}\neq\Bo$ (since the first measurement disturbs the system). 

The crucial point is that the symmetry properties of $\ii$ and $\Bo$ guarantee that the obtained measurement outcome distributions have the desired symmetries. 
Indeed, $\Co$ satisfies
\begin{equation*}
U_g V_{h} \Co (X\times Y) V_{h}^\ast U_g^\ast = \Co (gX \times hY)
\end{equation*}
for all $X\in \bor{G}$, $Y\in \bor{H}$, $g\in G$, $h\in H$. In particular, its marginal observables $\widetilde{\Ao}$ and $\widetilde{\Bo}$ have the same symmetry properties \eqref{eq:A-covar}-\eqref{eq:B-covar} as $\Ao$ and $\Bo$, respectively. 
For the first marginal $\widetilde{\Ao}$ we get
\begin{equation}\label{eq:A-unsharp-covar}
U_g \widetilde{\Ao}(X) U^\ast_g = \widetilde{\Ao}(gX) \qquad V_{h} \widetilde{\Ao}(X) V^\ast_{h} = \widetilde{\Ao}(X)
\end{equation}
for all $X\in\bor{G}$, $g\in G$, $h\in H$, and for the second marginal $\widetilde{\Bo}$ we get
\begin{equation}\label{eq:B-unsharp-covar}
U_g \widetilde{\Bo}(Y) U^\ast_g = \widetilde{\Bo}(Y) \qquad V_{h} \widetilde{\Bo}(Y) V^\ast_{h} = \widetilde{\Bo}(hY)
\end{equation}
for all $Y\in\bor{H}$, $g\in G$, $h\in H$.

In summary, we have seen that one simple condition \eqref{eq:ins-UV} on instruments leads to favorable symmetry properties on the obtained measurement outcome distributions.
The relevant questions are whether instruments satisfying the condition \eqref{eq:ins-UV} actually exist, and what kind of observables $\Co$ we can realize by formula \eqref{eq:G-seq}.
In the following we will answer these questions in the case of conjugated observables.

%%%%%%%%%%%%%%%%%%%%%%%%%%%%%%%%%%%%%%%%%%%%%%%%%%%%%%%%%%%
\subsection{Canonically conjugated observables and their approximate versions}\label{sec:receipt-conjugate}
%%%%%%%%%%%%%%%%%%%%%%%%%%%%%%%%%%%%%%%%%%%%%%%%%%%%%%%%%%%

In Subsec. \ref{sec:setup} the symmetry groups $G$ and $H$ had no connection.
We will now add more structure and assume that the situation is more specific.
This then leads to the usual notion of canonically conjugated observables.

Suppose $G$ is an abelian lcsc group and $\hG$ is its dual group. 
We denote additively the product in $G$ and multiplicatively the product in $\hG$. 
The pairing between $x\in G$ and $\chi\in\hG$ is the complex number $\chi(x)$.

\begin{definition}
A {\em Weyl system} for the pair $(G,\hG)$ is a couple of unitary representations $(U,V)$ of $G$ and $\hG$, respectively, defined on the same Hilbert space $\hh$ and satisfying
\begin{equation}\label{WHcommutaz}
U_x V_\chi = \overline{\chi(x)} V_\chi U_x
\end{equation}
for all $x\in G$, $\chi\in\hG$.
\end{definition}

We recall that, if $\Ao$ and $\Bo$ are sharp observables based on $G$ and $\hG$, respectively, and both with values in $\lh$, then by SNAG theorem formulas
\begin{equation}\label{eq:SNAG}
U_x = \int \overline{\chi(x)} \de \Bo (\chi) \, , \qquad V_\chi = \int \chi(x) \de \Ao (x)
\end{equation}
define two unitary representations $U$ of $G$ and $V$ of $\hG$ in the Hilbert space $\hh$, and, conversely, any couple of unitary representations $(U,V)$ of $G$ and $\hG$ in $\hh$ arise in this way from a unique couple of sharp observables $(\Ao,\Bo)$ on $G$ and $\hG$, respectively (see e.~g.~\cite[Theorem 4.44]{CAHA95}).

\begin{definition}
A pair of sharp obsevables $(\Ao,\Bo)$, with $\Ao : \bor{G} \to \lh$ and $\Bo : \bor{\hG} \to \lh$, are {\em canonically conjugated} if the couple of representations $(U,V)$ defined in eqs.~\eqref{eq:SNAG} is a Weyl system.

If canonically conjugated observables $(\Ao,\Bo)$ and a Weyl system $(U,V)$ are connected in this way, we call them {\em associated}.
\end{definition}

Two Weyl systems $(U,V)$ on $\hh$ and $(U^\prime,V^\prime)$ on $\hh^\prime$ [resp., two canonically conjugated observables $(\Ao,\Bo)$ on $\hh$ and $(\Ao^\prime,\Bo^\prime)$ on $\hh^\prime$] are {\em equivalent} if there exists a unitary operator $W:\hh\to\hh^\prime$ intertwining the pair of representations $(U,V)$ and $(U^\prime,V^\prime)$ [resp., the pair of sharp observables $(\Ao,\Bo)$ and $(\Ao^\prime,\Bo^\prime)$]. By eqs.~\eqref{eq:SNAG}, SNAG theorem estabilishes an identification of the class of Weyl systems and the class of canonically conjugated observables, and such identification preserves equivalence. Stone-von Neumann theorem then asserts that there exists {\em exactely one} equivalence class of Weyl systems, or, alternatively, canonically conjugated pairs \cite{MackeyStone}.

As a consequence of eq.~\eqref{WHcommutaz}, if $(\Ao,\Bo)$ is a pair of canonically conjugated observables and $(U,V)$ is its associated Weyl system, then, for all $x\in G$, $\chi\in\hG$,
\begin{equation}\label{eq:covA}
U_x \Ao (X) U_x^\ast = \Ao (X+x) \, \qquad V_\chi \Ao (X) V_\chi^\ast = \Ao (X)
\end{equation}
for all $X\in\bor{G}$, and
\begin{equation}\label{eq:covB}
 V_\chi \Bo (Y) V_\chi^\ast = \Bo (\chi Y) \, \qquad U_x \Bo (Y) U_x^\ast = \Bo (Y) 
\end{equation}
for all $Y\in\bor{\hG}$.
This shows that a pair of canonically conjugated observables is a special instance of our general setting described in Subsection \ref{sec:setup}.

The setup described in Subsec.~\ref{sec:setup} shows that an instrument with suitable symmetry properties leads to a sequential implementation of some observables $\widetilde{\Ao}$ and $\widetilde{\Bo}$ with the same symmetry properties than $\Ao$ and $\Bo$. To describe this situation, we introduce the following definition.

\begin{definition}
A pair of obsevables $(\widetilde{\Ao},\widetilde{\Bo})$, with $\widetilde{\Ao} : \bor{G} \to \lh$ and $\widetilde{\Bo} : \bor{\hG} \to \lh$, are {\em conjugated} if there exists a Weyl system $(U,V)$ on $\hh$ such that the covariance and invariance relations \eqref{eq:covA} and \eqref{eq:covB} hold with $\Ao$ and $\Bo$ replaced by $\widetilde{\Ao}$ and $\widetilde{\Bo}$, respectively.

If $(\widetilde{\Ao},\widetilde{\Bo})$ and $(U,V)$ are connected in this way, we say that $(\widetilde{\Ao},\widetilde{\Bo})$ are {\em related} to the Weyl system $(U,V)$, or, equivalently, to the canonically conjugated observables $(\Ao,\Bo)$ associated to $(U,V)$.
\end{definition}

Conjugated observables $(\widetilde{\Ao},\widetilde{\Bo})$ can be expressed in a very simple form in terms of a canonical conjugated pair $(\Ao,\Bo)$ related to them. Indeed, there exist probability measures $\sigma$ on $G$ and $\tau$ on $\hG$ such that $\widetilde{\Ao} \equiv \Ao_\sigma$ and $\widetilde{\Bo} \equiv \Bo_\tau$, where
\begin{equation}\label{eq:Afuzzy}
\Ao_\sigma (X) = \int \sigma(X-x) \de \Ao(x) \quad \forall X \in\bor{G}  \, ,
\end{equation}
and
\begin{equation}\label{eq:Bfuzzy}
\Bo_\tau (\hY) = \int \tau (\chi^{-1}\hY) \de \Bo(\chi) \quad \forall \hY \in\bor{\hG}  \, .
\end{equation}
This result has been proved in \cite{CaHeTo04} in the case $G=\R^n$, and the extension of the proof to the general setting (i.~e.~$G$ is a lcsc abelian group) is straightforward. If $\sigma = \delta_0$ [resp., $\tau = \delta_1$] is the Dirac measure centered at $0$ [resp., $1$], then $\Ao_\sigma = \Ao$ [resp., $\Bo_\tau = \Bo$]. The deviation of $\sigma$ and $\tau$ from Dirac $\delta$'s can be interpreted as imprecision or noise in the measurement of $\Ao$ and $\Bo$. 

We sometimes need to fix a concrete representation for conjugated and canonically conjugated observables.
 This representation also demonstrates that an equivalence class of canonically conjugated observables exists for any abelian lcsc group $G$. 
By Stone-von Neumann uniqueness theorem, fixing the representation does not affect the full generality of our discussion and results.
To start with, we fix Haar measures $\lam$ and $\hlam$ in $G$ and $\hG$, respectively, and set $\hh=\ldue{G,\lam} \equiv \ldue{G}$ and $\hhat=\ldue{\hG,\hlam} \equiv \ldue{\hG}$. 
There is then a unique real constant $c>0$ (depending on the choices of $\lam$ and $\hlam$) such that the Fourier transform
\begin{equation*}
\left( \ff f \right) (\chi) = c \int \overline{\chi (x)} f(x) \de\lambda (x) \qquad f\in \luno{G} \cap \ldue{G}
\end{equation*}
extends to a unitary map $\ff$ from $\hh$ to $\hhat$.
We take $U$ and $V$ to be the left regular representations of $G$ and $\hG$, hence acting in $\hh$ as
\begin{eqnarray*}
U_x f (y) & = & f (y-x) \quad \forall x\in G,\, f\in \hh \\
V_\chi f (y) & = & \chi(y) f(y) \quad \forall \chi\in \hG,\, f\in \hh \, .
\end{eqnarray*}
It is immediately checked that $(U,V)$ is a Weyl system. 
Its associated pair of canonically conjugated observables $(\Ao,\Bo)$ are given by
\begin{equation}\label{eq:A}
\Ao(X) f (x) = 1_X(x) f (x) \quad \forall X\in\bor{G}, \, f\in\hi \, , 
\end{equation} 
where $1_X$ is the characteristic function of a set $X$, and
\begin{equation}\label{eq:B}
[ \ff\Bo(Y)\ff^{-1} \hat{f} ] (\chi) = 1_\hY (\chi) \hat{f} (\chi) \quad \forall \hY\in\bor{\hG}, \, \hat{f}\in\hhat \, .
\end{equation}
If $(\Ao_\sigma,\Bo_\tau)$ are conjugated observables related to $(\Ao,\Bo)$ as in eqs.~\eqref{eq:Afuzzy}-\eqref{eq:Bfuzzy}, then
\begin{equation*}
[\Ao_\sigma (X) f] (x) = \sigma (X-x) f(x) \quad \forall f\in\hh ,\, X\in\bor{G}
\end{equation*}
and
\begin{equation*}
[\Bo_\tau (Y) f] (x) = c\, [(\ff^{-1} \tau_Y) \ast f] (x) \quad \forall f\in\hh \textrm{ and } Y\in \bor{\hG} \textrm{ such that } \hlam (Y) < \infty \, ,
\end{equation*}
where $\ast$ is the convolution and $\tau_Y (\chi) = \tau(\chi^{-1} Y)$ $\forall \chi\in\hG$.

%%%%%%%%%%%%%%%%%%%%%%%%%%%%%%%%%%%%%%%%%%
\subsection{Weyl-Heisenberg group}\label{sec:weyl}
%%%%%%%%%%%%%%%%%%%%%%%%%%%%%%%%%%%%%%%%%%

As explained earlier, we are interested on instruments satisfying the condition \eqref{eq:ins-UV} for a Weyl system $(U,V)$. 
This condition is not a single covariance condition (as there are two groups $G$ and $\hG$ involved), and for this reason it is convenient to introduce the {\em Weyl-Heisenberg group $\WH$ associated to $G$}.
The Weyl-Heisenberg group $\WH$ is the topological product $G\times \hG \times \T$ ($\T =$ the complex numbers with modulus $1$) endowed with the composition law
\begin{equation*}
(x,\chi,u)(y,\gamma,v) = (x+y,\chi\gamma, \overline{\chi(y)} uv) \, .
\end{equation*}
This makes $\WH$ a non-abelian topological group. Its centre is the subgroup $Z=\{ (0,1,u) \mid u\in \T \}$.

The abelian subgroup $N \equiv \{0\}\times \hG\times \T$ is normal and closed in $\WH
$, and the homogeneous space $\WH/N$ can be identified with $G$.
The group $\WH$ acts on the homogenous space $\WH/N$ in the usual way.
With the identification $\WH/N \simeq G$, the action of an element $( x,\chi,u )\in\WH$ on $y\in G$ is simply
$$
( x,\chi,u ) [y] = x+y \ .
$$

Setting
\begin{equation*}
W(x,\chi,u) = \overline{u} U_x V_\chi \, ,
\end{equation*}
we obtain an irreducible unitary representation of $\WH
$ in $\hh$, which is called the {\em Schr\"odinger representation}.
We refer to Chapter 1, \S 3 in \cite{HAPS89} for more details on such representation in the case $G=\hG = \R^n$, and to \cite{MackeyStone} for the general case.
We can now conclude that an instrument $\ii$ on $G$ satisfies eq.~\eqref{eq:ins-UV} if and only if 
\begin{equation*}
\ii_{(x,\chi,u) [X]}(\varrho) = W(x,\chi,u) \ii_X \left[ W(x,\chi,u)^\ast \varrho W(x,\chi,u) \right] W(x,\chi,u)^\ast
\end{equation*}
for all $X\in \bor{G}$, $(x,\chi,u) \in \WH$.
This is a single covariance condition in the sense defined by Davies \cite{Davies70}. 
Therefore, to understand the sequential measurements of conjugated observables, we need to study \emph{$W$-covariant instruments based on $\WH/N$}.

%%%%%%%%%%%%%%%%%%%%%%%%%%%%%%%%
\section{Covariant  instruments}\label{sec:structure}
%%%%%%%%%%%%%%%%%%%%%%%%%%%%%%%%

In this section we characterize the structure of $W$-covariant instruments based on $\WH/N$.
Our characterization of $W$-covariant instruments is best approached by first recalling a special class of measurement models (Subsec. \ref{sec:standard}).
The main structure theorem (Subsec. \ref{sec:theorem}) can then be seen as a natural extension of this model.
Based on these results, we will draw conclusions on the implementation of covariant phase space observables (Subsec. \ref{sec:cpso}).

%%%%%%%%%%%%%%%%%%%%%%%%%%%%%%%%%%%%%%%
\subsection{Von Neumann's measurement model}\label{sec:standard}
%%%%%%%%%%%%%%%%%%%%%%%%%%%%%%%%%%%%%%%

In his famous book \cite{MFQM55}, von Neumann described a position measurement scheme.
It is known that von Neumann's measurement model leads to a covariant instrument \cite{Ozawa86} (see also \cite{OQP97}, \cite{CaHeTo08}).
For this reason, it seems useful to have a closer look on it.

In our case, we intend to measure the observable $\Ao$ (defined in eq.~\eqref{eq:A}) by suitably coupling the system with an ancillary copy of it and measuring the observable $\Ao$ of the copy.
If we follow the idea of von Neumann's model, the measurement coupling on the composite system is described by the unitary operator $L: \hh\otimes \hh \to \hh\otimes \hh$,
\begin{equation}\label{eq:def:L}
Lf (x,y) = f (x,y-x) \quad \forall f\in \hh\otimes \hh = \ldue{G\times G, \lam \otimes \lam} \, ,
\end{equation}
which can be alternatively written as
\begin{equation*}
L = \int \overline{\chi (x)} \de (\Ao \otimes \Bo) (x,\chi) 
\end{equation*}
(for the standard identification $\hh\otimes \hh = \ldue{G\times G, \lam \otimes \lam}$ see e.~g.~\cite[Theorem 7.16]{CAHA95}).
In the case $G =\hG=\R$, with pairing $\chi_p (x) = e^{ipx}$, $p,x\in\R$, the last formula can be rewritten in the exponential form
$$
L = e^{-iQ\otimes P}
$$
with 
$$
Q=\int_{\R} x \de \Ao (x) \, , \qquad P=\int_{\R} p \de \Bo (\chi_p)
$$
the standard position and momentum selfadjoint operators,
thus showing the connection to von Neumann's measurement model.
For later purposes it is useful to note that the unitary operator $L$ satisfies the intertwining properties
\begin{equation}\label{eq:propL}
L (U_x \otimes U_y) = (U_x \otimes U_{x+y}) L \, , \quad L(V_\chi \otimes V_\gamma) = (V_{\chi \gamma^{-1}} \otimes V_\gamma) L
\end{equation}
for all $x,y\in G$ and $\chi,\gamma\in\hG$.

We choose the pointer observable to be $\Ao$ on the probe system.
Hence, if the initial state of the probe system is $\omega$, then the instrument $\ii^{\omega}$ deriving from our measurement model is
\begin{equation}\label{eq:standard-ins}
\ii^{\omega}_X (\varrho) = \trped{2}{\id\otimes \Ao(X) L \left( \varrho \otimes \omega \right) L^\ast} \, ,
\end{equation}
where ${\rm tr}_2 : \trhh \to \trh$ is the partial trace in $\trhh$ with respect to the second factor (the probe system). Using the intertwining properties \eqref{eq:propL}
 of $L$, the covariance properties \eqref{eq:A-covar} of $\Ao$ and the cyclicity of ${\rm tr}_2$ with respect to the second factor in the tensor product, one can check that $\ii^{\omega}$ is a $W$-covariant instrument on $G$.

With different choices of the probe state $\omega$ we can realize different instruments $\ii^{\omega}$.
There are also two ways to construct new $W$-covariant instruments from those of the form $\ii^{\omega}$.
First, if we fix $x\in G$, then the translated instrument
\begin{equation*}
X \mapsto U^\ast_x \ii^\omega_X(\cdot) U_x
\end{equation*}
is still $W$-covariant.
Another observation is that the set of $W$-covariant instruments is convex, so for collections of states $\omega_1,\ldots, \omega_n$, group elements $x_1,\ldots,x_n$ and positive numbers $t_1,\ldots,t_n$, $\sum_j t_j=1$, we get a $W$-covariant instrument
\begin{equation}\label{eq:convdisc}
X\mapsto \sum_{j=1}^n t_j  U^\ast_{x_j} \ii_X^{\omega_j}(\cdot) U_{x_j} \, .
\end{equation}
In the next subsection we will see that indeed \emph{every} $W$-covariant instrument arises in this way, possibly replacing the above finite sum with a suitably defined integral.

%%%%%%%%%%%%%%%%%%%%%%%%%%%%%%%%%%%%%%
\subsection{Structure of $W$-covariant instruments}\label{sec:theorem}
%%%%%%%%%%%%%%%%%%%%%%%%%%%%%%%%%%%%%%

Before stating the structure theorem for $W$-covariant instruments, we need to fix some additional mathematical notation. 
We recall that a $\trh$-valued {\em vector measure} on $G$ is a countably additive mapping $\Mo : \bor{G} \to \trh$ with finite total variation (we refer to \cite{Lang} for details). 
A $\trh$-valued vector measure $\Mo$ is {\em positive} if $\Mo (X)\geq O$ for all $X\in\bor{G}$, and {\em normalized} if $\tr{\Mo (G)} = 1$. 
We denote by $\mm(G;\trh)$ the linear space of $\trh$-valued vector measure on $G$, and by $\mm(G;\trh)_1$ the convex subset of positive normalized elements in $\mm(G;\trh)$.

A mapping $\Mo : \bor{G} \to \trh$ is a vector measure if and only if there exists a positive measure $\mu$ on $G$ and a map $M\in\luno{G,\mu; \trh}$ ($=$ the space of functions $M : G \to \trh$ which are $\mu$-integrable in the sense of Bochner) such that 
\begin{equation*}
\Mo(X) = \int 1_X (x) M(x) \de \mu (x)
\end{equation*}
for all $X\in\bor{G}$ \cite{Lang}. 
If $\Mo$ and the couple $(\mu,M)$ are related in this way, we write $\de\Mo(x) = M(x) \de \mu(x)$. 
Clearly, the correspondence $\Mo \leftrightarrow (\mu,M)$ is one-to-many.

For each positive measure $\mu$ on $G$, we denote by $\luno{G,\mu; \trh}_1$ the convex subset of elements in $\luno{G,\mu; \trh}$ such that $M(x)\geq O$ for $\mu$-a.a.~$x$ and $\int \tr{M(x)} \de\mu (x) = 1$. If $\Mo\in\mm(G;\trh)$ and $\de\Mo(x) = M(x) \de \mu(x)$, then $\Mo\in\mm(G;\trh)_1$ if and only if $M\in\luno{G,\mu; \trh}_1$.

With this preparation, we are now ready to state our main structure theorem.

\begin{theorem}\label{th:gen}
There is a convex one-to-one correspondence between $\mm(G;\trh)_1$ and the set of $W$-covariant instruments. 
If $\Mo\in\mm(G;\trh)_1$, with $\de\Mo (x) = M(x) \de\mu (x)$, the corresponding instrument $\ii$ is given by
\begin{equation}\label{eq:gen}
\ii_X (\varrho) = \trped{2}{(\id\otimes \Ao(X)) L \left[ \int U_x^\ast \varrho U_x \otimes M(x) \de\mu (x) \right] L^\ast} ,
\end{equation}
where $L: \hh\otimes \hh \to \hh\otimes \hh$ is the unitary operator defined in eq.~\eqref{eq:def:L}
and the integral in eq.~\eqref{eq:gen} is defined in $\trhh$ in the sense of Bochner.
\end{theorem}

The proof of Theorem \ref{th:gen} requires some preliminary technical results, therefore we postpone it to Section  \ref{app:proof:th:gen}.

If $M\in\luno{G,\mu;\trh}_1$ and we define $M^\prime (x) = U_x^\ast M(x) U_x$, then still $M^\prime \in\luno{G,\mu;\trh}_1$ and the $\trh$-valued measure $\Mo^\prime$ with $\de\Mo^\prime (x) = M^\prime (x) \de\mu (x)$ still belongs to $\mm(G;\trh)$. Formula~\eqref{eq:gen} becomes
\begin{equation}\label{eq:genbis}
\ii_X (\varrho) = \int U_x^\ast\ \trped{2}{(\id\otimes \Ao(X)) L (\varrho \otimes M^\prime (x)) L^\ast} U_x \de\mu (x)
\end{equation}
by the intertwining properties \eqref{eq:propL} of $L$. 
We can further normalize $M^\prime (x)$ to the state $\omega(x)= M^\prime (x)/\no{M^\prime (x)}_1$ (with the convention $0/0 = 0$), define the probability measure $\de\nu (x) = \no{M^\prime (x)}_1 \de \mu (x)$, and then eq.~\eqref{eq:genbis} can be rewritten as
\begin{equation}\label{eq:genterza}
\ii_X (\varrho) = \int U_x^\ast \lfq \ii^{\omega(x)}_X (\varrho) \rgq U_x \de\nu (x) \, ,
\end{equation}
where each $\ii^{\omega(x)}$ is an instrument of the standard form \eqref{eq:standard-ins}.
Comparing eqs.~\eqref{eq:convdisc} and \eqref{eq:genterza} we see that $\ii$ is a continuous convex combination of translations of von Neumann-type instruments $\ii^{\omega(x)}$.

As discussed in Subsec.~\ref{sec:receipt-conjugate}, the actually measured observable (the one associated to $\ii$) is not $\Ao$ but its unsharp version $\tilde{\Ao}$ defined in eq.~\eqref{eq:marg-A}. Similarly, the observable $\tilde{\Bo}$ defined in eq.~\eqref{eq:marg-B} is an unsharp version of $\Bo$. Since by eqs.~\eqref{eq:A-unsharp-covar} and \eqref{eq:B-unsharp-covar} the pair of observables $(\tilde{\Ao},\tilde{\Bo})$ are conjugated and related to the canonically conjugated observables $(\Ao,\Bo)$, we have $(\tilde{\Ao},\tilde{\Bo}) = (\Ao_\sigma,\Bo_\tau)$ for some probability measures $\sigma$ and $\tau$. 
The relation of $\sigma$ and $\tau$ to Theorem \ref{th:gen} is the following.

\begin{proposition}\label{prop:dens}
If $\ii$ is the instrument defined by eq.~\eqref{eq:genbis}, and $(\tilde{\Ao},\tilde{\Bo})$ is the pair of observables defined in eqs.~\eqref{eq:marg-A} and \eqref{eq:marg-B}, then $(\tilde{\Ao},\tilde{\Bo}) = (\Ao_\sigma,\Bo_\tau)$, where the probability measures $\sigma$ and $\tau$ are given by
\begin{eqnarray}
\sigma(X) & = & \tr{\Ao(X) \Mo^\prime (G)} \label{eq:dens-A} \\
\tau(Y) & = & \tr{\Bo(Y^{-1}) \Mo^\prime (G)} \label{eq:dens-B}
\end{eqnarray}
for all $X\in\bor{G}$, $Y\in\bor{\hG}$.
\end{proposition}

The proof of Proposition \ref{prop:dens} is given in Section \ref{app:proof:th:gen}.

%%%%%%%%%%%%%%%%%%%%%%%%%%%%%%%%%%%%%%
\subsection{Covariant phase space observables}\label{sec:cpso}
%%%%%%%%%%%%%%%%%%%%%%%%%%%%%%%%%%%%%%

The centre of the Weyl-Heisenberg group $\WH$ is $Z = \{0\}\times \{1\}\times \T$, and the homogeneous space $\WH/Z$ can be identified with $G\times \hG$.
The action of an element $(x,\chi,u)\in \WH$ on $(y,\gamma)\in G\times\hG$ is then
$$
(x,\chi,u)[(y,\gamma)] = (x+y , \chi\gamma) \, .
$$
The identification $\WH/Z\simeq G\times \hG$ is used in the following formulation of covariant phase space observables.

Suppose that $\ii$ is a $W$-covariant instrument based on $G$.
As we have seen earlier, the observable $\Co$ defined by
\begin{equation}\label{eq:seq-ph-sp-obs}
\Co (X\times Y) = \ii^\ast_X (\Bo (Y)) \quad \forall X\in\bor{G},\, Y\in\bor{\hG} ,
\end{equation}
satisfies
\begin{equation}
U_x V_{\chi} \Co (X\times Y) V_{\chi}^\ast U_x^\ast = \Co ((x+X) \times \chi Y)
\end{equation}
for all $X\in \bor{G}$, $Y\in \bor{\hG}$, $x\in G$, $\chi\in \hG$.
We can rewrite this condition as
\begin{equation}
W(x,\chi,u) \Co (X\times Y) W(x,\chi,u)^\ast = \Co ((x,\chi,u) [X \times Y]) \, .
\end{equation}
This equation extends to the whole $\sigma$-algebra $\bor{G\times\hG}$ generated by product sets, hence
\begin{equation}\label{eq:cov-ph-sp-obs}
\Co((x,\chi,u) [Z]) = W(x,\chi,u) \Co(Z) W(x,\chi,u)^\ast
\end{equation}
for all $Z\in \bor{G\times\hG}$ and $(x,\chi,u) \in \WH$.
We will call an observable $\Co$ satisfying eq.~\eqref{eq:cov-ph-sp-obs} a \emph{covariant phase space observable}.

The set of covariant phase space observables is in one-to-one correspondence with the set $\sh$ \cite{cpso}.
If $S\in\sh$, then the corresponding covariant phase space observable $\Co_S$ is given by
\begin{equation*}
\Co_S(Z) = c^2 \int_Z U_x V_\chi S V_\chi^\ast U_x^\ast \de (\lambda \otimes \hlam)(x,\chi) \quad \forall Z\in\bor{G\times \hG} \, .
\end{equation*}
In particular, if $S$ is a one-dimensional projection $S=\kb{\eta}{\eta}$, then the observable $\Co_S$ is generated by a family of  generalized coherent states $\{ U_x V_\chi\eta \mid (x,\chi)\in\WH/Z \}$.

\begin{proposition}\label{prop:di-S}
Let $\ii$ be a $W$-covariant instrument defined in eq.~\eqref{eq:genbis} and $\Co$ the covariant phase space observable defined in eq.~\eqref{eq:seq-ph-sp-obs}. 
Then $\Co = \Co_S$, with $S\in \sh$ determined by condition
\begin{equation}\label{eq:di-S}
\ip{f_1}{S f_2} = \ip{\check{f}_2}{\Mo{^\prime (G)} \check{f}_1} \quad \forall f_1,f_2 \in\hh \, ,
\end{equation}
where $f\mapsto \check{f}$ is the antiunitary mapping on $\hh=\ldue{G}$ given by
$$
\check{f} (x) = \overline{f(-x)} .
$$
\end{proposition}

Again, we postpone the proof of Proposition \ref{prop:di-S} to Section \ref{app:proof:th:gen}.

If $T\in\trh$, then defining $\check{T}$ as
\begin{equation}\label{eq:omegacheck}
\ip{f_1}{\check{T} f_2} = \ip{\check{f}_2}{T \check{f}_1} \quad \forall f_1,f_2 \in\hh \, ,
\end{equation}
we obtain a linear isometric isomorphism of $\trh$ with itself. 
This follows easily since, if $T = \sum_i \lam_i \kb{v_i}{u_i}$ is the singular values decomposition of $T$, then $\check{T} = \sum_i \lam_i \kb{\check{u}_i}{\check{v}_i}$ is the singular values decomposition of $\check{T}$.
Therefore, Proposition \ref{prop:di-S} leads to the following conclusion.

\begin{corollary}\label{cor:every-cpso-is-seq}
Every covariant phase space observable has a sequential implementation of the form \eqref{eq:seq-ph-sp-obs}.
\end{corollary}

Let us notice that in formula \eqref{eq:di-S} the $\trh$-valued measure $\Mo'$ occurs only through its total value $\Mo{^\prime (G)}$.
This means, in particular, that each covariant phase space observable has a sequential implementation where the covariant instrument is of the standard form \eqref{eq:standard-ins}.
Moreover, the correspondence $\omega \leftrightarrow S$ between the probe states $\omega$ and the generating operators $S$ is then one-to-one and given by $\omega = \check{S}$.
We can thus state our result in the following form.

\begin{corollary}\label{cor:every-cpso-can-be-programmed}
All covariant phase space observables can be implemented with the same measurement coupling and pointer observable just by changing the probe state. 
\end{corollary}

We recall that a \emph{(measurement-assisted) programmable quantum processor} is a measurement process where the initial probe state can be changed \cite{ZiBu05}.
The initial probe state is thought as a program that encodes different observables.
We can translate Corollary \ref{cor:every-cpso-can-be-programmed} into the statement that \emph{all covariant phase space observables can be implemented on a single programmable quantum processor}.

%%%%%%%%%%%%%%%%%%%%%%%%%%%%%%%%%%%%%%%%%%%%%%%%
\section{Examples}\label{sec:examples}
%%%%%%%%%%%%%%%%%%%%%%%%%%%%%%%%%%%%%%%%%%%%%%%%

In Subsections \ref{sec:position} and \ref{sec:spin} below we illustrate the results in the concrete cases of position-momentum and orthogonal spin components.
In particular, we show connections to some earlier studies.

%%%%%%%%%%%%%%%%%%%%%%%%%%%%%%%%%%%%%%%%%%%%%%%
\subsection{Sequential measurement of position and momentum}\label{sec:position}
%%%%%%%%%%%%%%%%%%%%%%%%%%%%%%%%%%%%%%%%%%%%%%%%

Let us consider a particle moving in the real line $\R$. Its associated Hilbert space is $\hh = \ldue{\R}$, where Haar measure $\lam$ on $\R$ is just Lebesgue measure.
Characteristic symmetry transformations for the particle include space translations and velocity boosts. 
These are described by two unitary representations $U$ (translations) and $V$ (boosts) of $\R$, acting on a vector $\psi\in\hi$ as
\begin{equation}
\left[U_q\psi\right](x) = \psi(x-q) \, , \qquad \left[V_p\psi\right](x) = e^{ipx}\psi(x) \, .
\end{equation}
Since $\hat{\R} = \R$, the couple $(U,V)$ is clearly a Weyl system for the pair $(\R,\R)$. Its associated canonically conjugated observables $(\Ao , \Bo)$ are just sharp position and momentum observables, respectively.

A translation covariant instrument for an unsharp position observable was introduced by Davies in \cite{QTOS76}. 
Here we show its connection to our general structure theorem. 
We do this by writing eq.~\eqref{eq:standard-ins}, the definition of $\ii^{\omega}$, in an alternative form. 
In the following we assume that the probe state $\omega$ is pure, hence $\omega = \kb{\varphi}{\varphi}$ for some fixed unit vector $\varphi\in\ldue{\R}$.

Let $\phi,\psi\in\hi$ be unit vectors, and $\varrho = \kb{\phi}{\phi}$. We have
\begin{eqnarray*}
\ip{\psi}{\ii_X^{\omega} (\varrho) \psi} & = & \tr{ \kb{\psi}{\psi} \ \trped{2}{\id\otimes \Ao(X) L ( \varrho \otimes \omega) L^\ast}} \\
& = & \tr{(\kb{\psi}{\psi} \otimes \Ao(X)) L ( \kb{\phi \otimes \varphi}{\phi \otimes \varphi} ) L^\ast} \\
& = & \ip{L(\phi \otimes \varphi)}{[\kb{\psi}{\psi} \otimes \Ao(X)] L(\phi \otimes \varphi)} \\
& = & \iint \overline{\phi(x) \varphi (y-x)} \psi(x) \left( \int \overline{\psi(z)} 1_X (y) \phi(z) \varphi(y-z) \de z \right) \de x \de y \\
& = & \int \overline{\psi(z)} \left[ \phi(z) \int \left( \int 1_X (y) \varphi(y-z) \overline{\varphi (y-x)} \de y \right) \overline{\phi(x)} \psi(x) \de x \right] \de z
\end{eqnarray*}
It follows that $\ii_X^{\omega}\left( \varrho \right)$ is the integral operator with kernel
\begin{equation*}
\Gamma (x,y) = \phi(x) \overline{\phi(y)} \int 1_X (u) \varphi(u-x) \overline{\varphi (u-y)} \de u \, ,
\end{equation*}
i.~e.
$$
[\ii_X^{\omega} (\varrho) \psi](x) = \int \Gamma (x,y) \psi(y) \de y
$$
If $\varphi$ is not only square integrable but also essentialy bounded, then the operator
\begin{equation*}
K_u :\ldue{\R} \rightarrow \ldue{\R} \, , \quad \left[ K_u \psi\right] (y) = \varphi (u-y) \psi(y) 
\end{equation*}
is well defined and the instrument $\ii^{\omega}$ takes the form
\begin{eqnarray*}
\ii_X^\omega(\varrho) & = & \int_X K_u \varrho K_u^\ast \de u \qquad \forall \varrho \in \sh \, .
\end{eqnarray*}
This way of writing is used e.g. in \cite{OQP97,QTOS76,Ozawa86} and we have thus seen that it arises from the general description under certain conditions.

%%%%%%%%%%%%%%%%%%%%%%%%%%%%%%%%%%%%%%%%%%%%%%%%
\subsection{Sequential measurement of two orthogonal spin components}\label{sec:spin}
%%%%%%%%%%%%%%%%%%%%%%%%%%%%%%%%%%%%%%%%%%%%%%%%

Let us consider a sequential measurement of two orthogonal spin components of a spin-$\half$ system.
The associated Hilbert space is $\hh = \C^2$.
A measurement outcome in each spin component measurement is either $+1$ (up) or $-1$ (down), hence we set $G=\Z_2=\{+ 1, -1\}$.
The two sharp observables that we intend to measure are
\begin{equation}
\So^\va(\pm 1) = \half \left( \id \pm \va\cdot\vsigma \right) \, , \quad \So^\vb(\pm 1) = \half \left( \id \pm \vb\cdot\vsigma \right) \, ,
\end{equation} 
where $\va$ and $\vb$ are orthogonal unit vectors in $\R^3$, and $\vsigma = (\sigma_1 ,\sigma_2 ,\sigma_3)$ is the vector consisting of Pauli operators.

The covariance properties of these observables are formulated in terms of two representations $U$ and $V$ of $\Z_2=\{+ 1 , -1 \}$, which correspond to the $180^\circ$-rotations with axis in the directions of $\vb$ and $\va$, respectively. 
Therefore, these representations are given by
\begin{equation*}
U_+ =\id \, , \,  U_- =\vb\cdot\vsigma \, , \quad V_+ =\id \, , \, V_-=\va\cdot\vsigma \, , 
\end{equation*} 
and the covariance properties are
\begin{equation*}
U_-\So^\va(j)U^\ast_- = \So^\va(-j) \, , \qquad V_-\So^\va(j)V^\ast_- = \So^\va(j)
\end{equation*} 
and
\begin{equation*}
U_-\So^\vb(j)U^\ast_- = \So^\vb(j) \, , \qquad V_-\So^\vb(j)V^\ast_- = \So^\vb(-j) .
\end{equation*}
It is straightforward to check that $(U,V)$ is the Weyl system for the pair $(\Z_2 , \Z_2)$ and its associated canonically conjugated observables are just $(\So^\va , \So^\vb)$.

Let $\{ e^\va_+ , e^\va_- \}$ be an orthonormal basis of $\C^2$ diagonalizing the operator $\va \cdot \vsigma$, with eigenvalues $1,-1$ respectively. 
We choose the phases of $e^\va_{\pm}$ such that $\vb\cdot\vsigma e^\va_{\pm}=e^\va_{\mp}$.
We can then write $\So^\va (k) = \kb{e^\va_k}{e^\va_k}$ and $\So^\vb (h) = \ff \So^\va (h) \ff = \frac{1}{2} \kb{e^\va_+ + h e^\va_-}{e^\va_+ + h e^\va_-}$, where
the Fourier transform $\ff:\C^2\to\C^2$ reads
\begin{equation*}
\ff (\alpha e^\va_+ + \beta e^\va_-) = \frac{\alpha + \beta}{\sqrt{2}} e^\va_+ + \frac{\alpha - \beta}{\sqrt{2}} e^\va_- \, .
\end{equation*}

By Theorem \ref{th:gen}, any $W$-covariant instrument $\ii$ on $\Z_2$ is given by
\begin{equation*}
\ii_k (T) = \sum_{x\in \Z_2} \trped{2}{(I\otimes \So^\va (k)) L (U_x^\ast T U_x \otimes M_x) L^\ast}
\end{equation*}
for some positive operators $M_+$, $M_-$ satisfying $\tr{M_+} + \tr{M_-} = 1$.
The unitary operator $L : \C^2 \otimes \C^2 \to \C^2 \otimes \C^2$ has the form
$$
L(e^\va_i \otimes e^\va_k) = e^\va_i \otimes e^\va_{ik} \quad \forall i,k\in\Z_2 \, ,
$$
and we notice that $L^2=\id$ (hence $L^\ast=L$).

For our purposes, it is enough to analyze the standard form instruments $\ii^{\omega}$ (i.~e.~$M_-=0$, $M_+\equiv \omega$) since all the covariant instruments are convex combinations of these and their translates.
We get
\begin{eqnarray*}
\ip{e^\va_i}{\ii^\omega_k (\varrho)e^\va_j} & = & \ip{e^\va_i}{\trped{2}{(\id\otimes \So^\va (k)) L (\varrho \otimes \omega ) L} e^\va_j} \\
& = & \tr{\kb{e^\va_j}{e^\va_i}\ \trped{2}{(\id\otimes \So^\va (k)) L (\varrho \otimes \omega) L}} \\
& = & \tr{\kb{e^\va_j \otimes e^\va_k}{e^\va_i \otimes e^\va_k}\ L ( \varrho \otimes \omega) L} \\
& = & \tr{\kb{e^\va_j \otimes e^\va_{jk}}{e^\va_i \otimes e^\va_{ik}}\ \varrho \otimes \omega } \\
& = & \ip{e^\va_{i}}{\varrho  e^\va_{j}} \ip{e^\va_{i}}{U_k \omega U_k e^\va_{j}} \, .
\end{eqnarray*}
In other words, the matrix $\ip{e^\va_i}{\ii^{\omega}_k (\varrho)e^\va_j}$ is just the Kronecker product of the matrices $\ip{e^\va_{i}}{\varrho e^\va_{j}}$ and $\ip{e^\va_{i}}{U_k \omega U_k e^\va_{j}}$.

By Proposition \ref{prop:dens}, the conjugated observables $(\tilde{\So}^\va, \tilde{\So}^\vb)$ corresponding to $\ii^\omega$ as in eqs.~\eqref{eq:marg-A} and \eqref{eq:marg-B} are characterized by the probability distribution
\begin{eqnarray*}
\sigma(k) & = & \tr{\omega \So^\va(k)} \\
\tau(k) & = & \tr{\omega \So^\vb(k)}
\end{eqnarray*}
for all $k\in\Z_2$.
Denoting $s=2\sigma(1)-1$ and $t=2\tau(1)-1$, we can write the observables $(\tilde{\So}^\va, \tilde{\So}^\vb)$ in the form
\begin{eqnarray*}
\So^{s\va}(\pm 1) & := & \tilde{\So}^\va (\pm 1) = \half \left( \id \pm s\va\cdot\vsigma \right) \\
\So^{t\vb}(\pm 1) & := & \tilde{\So}^\vb (\pm 1) = \half \left( \id \pm t\vb\cdot\vsigma \right) .
\end{eqnarray*}
These are recognized as \emph{unsharp spin observables}, first introduced in \cite{Busch86}.

If we write the state $\omega$ in the form $\omega=\half (\id + \vr\cdot\vsigma )$ for some vector $\vr\in\R^3$, $\no{\vr}\leq 1$, then $s=\vr\cdot\va$ and $t=\vr\cdot\vb$.
Since $\va$ and $\vb$ are orthogonal unit vectors, we obtain
\begin{equation*}
s^2+t^2 = (\vr\cdot\va)^2 + (\vr\cdot\vb)^2 \leq \no{\vr}^2 \leq 1 \, .
\end{equation*}
This trade-off relation between the accuracies of $\So^{s\va}$ and $\So^{t\vb}$ was derived in \cite{Busch86} from the assumption that  $\So^{s\va}$ and $\So^{t\vb}$ are jointly measurable (see also \cite{AnBaAs05} and \cite{BuHe08} for different type of derivations of the same relation).

\section{Proofs}\label{app:proof:th:gen}

We first recall some basic facts from the theory of integral operators and tensor products.

Suppose $\mu_1$ and $\mu_2$ are positive Borel measures on lcsc spaces $\Omega_1$ and $\Omega_2$, respectively. 
A linear operator $A: \ldue{\Omega_1,\mu_1} \frecc \ldue{\Omega_2,\mu_2}$ is an {\em integral operator} if there exists a measurable map $A : \Omega_2 \times \Omega_1 \frecc \C$ (the {\em kernel} associated to $A$) such that
$$
Af (x) = \int A(x;y) f(y) \de\mu_1 (y) \quad \forall f\in\ldue{\Omega_1,\mu_1}
$$
(we use the semicolon in the kernel to separate the variables referring to the $L^2$-spaces of the domain and the image of $A$).

We then have the following fact.
\begin{theorem}
If $\Omega$ is a lcsc space and $\mu$ is a positive Borel measure on $\Omega$, then a linear operator $T : \ldue{\Omega,\mu} \frecc \ldue{\Omega,\mu}$ is in the space $\ti{\ldue{\Omega,\mu}}$ of trace class operators on $\ldue{\Omega,\mu}$ if and only if
\begin{enumerate}
\item[\rm (i)] $T$ is an integral operator;
\item[\rm (ii)] there exists a lcsc space $\Omega^\prime$, a positive Borel measure $\mu^\prime$ on $\Omega^\prime$, and elements $A_1 , A_2 \in \ldue{\Omega \times\Omega^\prime , \mu \otimes \mu^\prime}$ such that
\begin{equation}\label{eq:intT}
T(x;y) = \int A_1 (x;z) \overline{A_2 (y;z)} \de\mu^\prime (z) \quad \textrm{for } \mu\otimes\mu \textrm{ - a.~a.~} (x,y) .
\end{equation}
\end{enumerate}
In this case, the trace of $T$ is
\begin{equation}\label{eq:trT}
\tr{T} = \int A_1 (x;y) \overline{A_2 (x;y)} \de (\mu\otimes\mu^\prime) (x,y) .
\end{equation}
Moreover, $T\in\ti{\ldue{\Omega,\mu}}_+$ if and only if one can choose $A_1 = A_2$ in eq.~\eqref{eq:intT}.
\end{theorem}

\begin{proof}
By Theorem VI.22 in \cite{MMMPI80}, $T$ is trace class if and only if there exist two Hilbert-Schmidt operators $A_1$ and $A_2$, with $A_i : \ldue{\Omega^\prime , \mu^\prime} \frecc \ldue{\Omega , \mu}$, such that $T = A_1 A_2^\ast$. By Theorem VI.23 in \cite{MMMPI80}, each $A_i$ is an integral operator with kernel $A_i \in \ldue{\Omega \times \Omega^\prime , \mu \otimes \mu^\prime}$, hence $T$ is an integral operator with kernel \eqref{eq:intT}. Moreover, $\tr{T} = \tr{A_1 A_2^\ast} = \scal{A_1}{A_2}_{HS}$, where $\scal{\cdot}{\cdot}_{HS}$ is the scalar product in the Hilbert space of Hilbert-Schmidt operators. By Theorem VI.23 in \cite{MMMPI80},
$$
\scal{A_1}{A_2}_{HS} = \int A_1 (x;y) \overline{A_2 (x,y)} \de (\mu \otimes \mu^\prime) (x,y) ,
$$
and eq.~\eqref{eq:trT} then follows.
\end{proof}

We use this characterization of trace class operators in the next four auxiliary lemmas.

\begin{lemma}\label{lem:1}
Let $\mu$ be a positive measure on the lcsc space $\Omega$, and $T$ an integral operator on $\ldue{\Omega,\mu}$. Then $T\in\ti{\ldue{\Omega,\mu}}$ if and only if there exists a Hilbert space $\vv$ and functions $\phi_1,\phi_2 \in\ldue{\Omega, \mu ; \vv}$ such that
\begin{equation}\label{eq:intT2}
T(x;y) = \scal{\phi_1 (x)}{\phi_2 (y)} \quad \textrm{for } \mu\otimes\mu \textrm{ - a.~a.~} (x,y) .
\end{equation}
In this case, the trace of $T$ is
$$
\tr{T} = \scal{\phi_1}{\phi_2}_{L^2} \, .
$$
Moreover, $T$ is positive if and only if one can choose $\phi_1 = \phi_2 = \phi$ in eq.~\eqref{eq:intT2}, and in this case
$$
\no{T}_1 = \no{\phi}^2_{L^2} \, .
$$
\end{lemma}

\begin{proof}
If $T$ is trace class, then choose $\mu^\prime$, $\Omega^\prime$, $A_1,A_2 \in \ldue{\Omega\times \Omega^\prime , \mu\otimes \mu^\prime}$ as in eq.~\eqref{eq:intT}, and let $\vv = \ldue{\Omega^\prime , \mu^\prime}$. Then the maps $\phi_i : x\mapsto A_i (x;\cdot)$, $i=1,2$, are in $\ldue{\Omega,\mu ; \vv}$ by Fubini theorem, and eq.~\eqref{eq:intT2} is just a rewriting of eq.~\eqref{eq:intT}.

Conversely, suppose $\vv$ is a Hilbert space and there exists $\phi_1,\phi_2 \in \ldue{\Omega,\mu ; \vv}$ for which eq.~\eqref{eq:intT2} holds. Let $\{ e_n \}_{n\in I}$ be an orthonormal basis in $\vv$ and define $A_i (x,n) = \scal{\phi_1 (x)}{e_n}$, $i=1,2$. Then $A_1, A_2 \in\ldue{\Omega \times \Omega^\prime, \mu\otimes\mu^\prime}$, where $\Omega^\prime = I$ and $\mu^\prime$ is the counting measure of the index set $I$, and with this choice $T$ satisfies eq.~\eqref{eq:intT}. It then follows from the previous discussion that $T\in\ti{\ldue{\Omega,\mu}}$ and
\begin{eqnarray*}
\tr{T} & = & \int \sum_n A_1 (x;n) \overline{A_2 (x;n)} \de \mu (x) = \int \scal{\phi_1 (x)}{\phi_2 (x)} \de \mu (x) \\
& = & \scal{\phi_1}{\phi_2}_{L^2} \, .
\end{eqnarray*}
The rest of the statement is also a straightforward consequence of the above considerations and of the fact that, if $T$ is positive, then $\no{T}_1 = \tr{T}$.
\end{proof}

If $f\in\ldue{\Omega,\mu}$, we introduce the notation $T_f (x;y) = f(x) \overline{f(y)}$. The associated integral operator $T_f = \kb{f}{f}$ on $\ldue{\Omega,\mu}$ is clearly trace class and positive.

\begin{lemma}\label{lem:1.5}
Let $\vv$ be a Hilbert space and $\mu_1 , \mu_2$ positive measures on the lcsc spaces $\Omega_1 , \Omega_2$, respectively. 
Suppose $\phi\in\ldue{\Omega_1 \times \Omega_2, \mu_1\otimes\mu_2 ; \vv}$. 
If $T$ is the integral operator on $\ldue{\Omega_1 \times \Omega_2, \mu_1\otimes\mu_2}$ with kernel
$$
T(x_1 , x_2 ; y_1 , y_2) = \scal{\phi(x_1 , x_2)}{\phi(y_1 , y_2)} \, ,
$$
then $T\in\ti{\ldue{\Omega_1 \times \Omega_2,\mu_1\otimes\mu_2}}_+ = \ti{\ldue{\Omega_1,\mu_1} \otimes \ldue{\Omega_2,\mu_2}}_+$, and its partial trace $\trped{2}{T}$ with respect to $\ldue{\Omega_2,\mu_2}$ is the integral operator on $\ldue{\Omega_1,\mu_1}$ with kernel
\begin{equation}\label{eq:ker-part-tr}
\lfq\trped{2}{T}\rgq (x_1 ; y_1) = \int \scal{\phi(x_1 , z)}{\phi(y_1 , z)} \de\mu_2 (z) \, .
\end{equation}
\end{lemma}

\begin{proof}
$T\in\ti{\ldue{\Omega_1\times\Omega_2,\mu_1\otimes\mu_2}}_+$ by Lemma \ref{lem:1}. If $f\in\ldue{\Omega_1,\mu_1}$, then
\begin{equation}\label{eq:1}
\scal{\trped{2}{T} \, f}{f} = \tr{(T_f\otimes \id_{\ldue{\Omega_2,\mu_2}})\, T} .
\end{equation}
The operator $(T_f\otimes \id_{\ldue{\Omega_2,\mu_2}})\, T$ is the integral operator with kernel
$$
\lfq (T_f\otimes \id_{\ldue{\Omega_2,\mu_2}})\, T \rgq] (x_1 , x_2 ; y_1 , y_2) = \scal{f(x_1) \int \overline{f(z)} \phi(z , x_2) \de\mu_1 (z)}{\phi(y_1 , y_2)} ,
$$
where it is easy to check that the map $(x_1 , x_2) \mapsto f(x_1) \int \overline{f(z)} \phi(z , x_2) \de\mu_1 (z)$ is in $\ldue{\Omega_1\times\Omega_2, \mu_1\otimes\mu_2 ; \vv}$ by an application of H\"older inequality and Fubini theorem. Then, by eq.~\eqref{eq:1} and Lemma \ref{lem:1},
\begin{eqnarray*}
\scal{\trped{2}{T} \, f}{f} & = & \int \scal{f(x_1) \int \overline{f(z)} \phi(z , x_2) \de\mu_1 (z)}{\phi(x_1 , x_2)} \de (\mu_1 \otimes \mu_2) (x_1 , x_2) \\
& = & \iint \lfq \int \scal{\phi(z , x_2)}{\phi(x_1 , x_2)} \de\mu_2 (x_2) \rgq f(x_1) \overline{f(z)} \de\mu_1 (x_1) \de\mu_1 (z) .
\end{eqnarray*}
This shows that $\trped{2}{T}$ is the integral operator in $\ldue{\Omega_1,\mu_1}$ with kernel \eqref{eq:ker-part-tr}, as claimed.
\end{proof}

We recall that $\hh = \ldue{G,\lam}$, where $\lambda$ is the Haar measure of $G$.
\begin{lemma}\label{lem:2}
Let $\vv$ be an Hilbert space and $\mu$ be a positive measure on $G$. Suppose $\phi_1,\phi_2\in\ldue{G\times G, \lam\otimes\mu ; \vv}$. Then there is an element $M\in\luno{G,\mu;\trh}$ such that
\begin{equation}\label{K<-->phi}
[M(h)](x;y) = \scal{\phi_1 (x,h)}{\phi_2 (y,h)} \quad \textrm{for } \mu\otimes\lam\otimes\lam \textrm{-a.~a.~} (h,x,y) .
\end{equation}

Moreover,
\begin{equation}\label{tr M(h)}
\tr{M(h)} = \int \scal{\phi_1 (x,h)}{\phi_2 (x,h)} \de\lam (x) \quad \textrm{for } \mu \textrm{-a.~a.~} h .
\end{equation}

If $\phi_1 = \phi_2 = \phi$ and $\no{\phi}_{L^2} = 1$, then $M\in\luno{G,\mu;\trh}_1$.
\end{lemma}
\begin{proof}
We prove the lemma for $\phi_1 = \phi_2 = \phi$, the general case following by polarization.

By Fubini theorem, there exists a $\mu$-null set $Z\in\bor{G}$ such that $\phi (\cdot,h)\in\ldue{G,\lam;\vv}$ for all $h\in G\setminus Z$. For such $h$'s, eq.~\eqref{K<-->phi} defines an element $M(h)\in\trh_+$, and
$$
\tr{M(h)} = \no{M(h)}_1 = \no{\phi (\cdot,h)}_{L^2} = \int \no{\phi (x,h)}^2 \de\lam (x)
$$
by Lemma \ref{lem:1}.

For all $f\in\hh$,
$$
\scal{M(h) f}{f} = \no{\int f(x) \phi(x,h) \de\lam (x)}^2 \quad \forall h\in G\setminus Z .
$$
Since the map $h\mapsto \int f(x) \phi(x,h) \de\lam (x)$ is measurable from $G$ into $\vv$ by Fubini theorem, the map $h\mapsto \scal{M(h) f}{f}$ is measurable. The map $h\mapsto \tr{AM(h)}$ is then measurable for every finite rank operator $A\in\lh$, hence is measurable for all $A\in\lh$ by sequential weak-* density of finite rank operators in $\lh$. By Corollary 2, p.~73 in \cite{Hille-Ph}, $M:G\frecc\trh$ is then a measurable map. Moreover,
$$
\int \no{M(h)}_1 \de\mu (h) = \iint \scal{\phi(x,h)}{\phi(x,h)} \de\lam (x) \de\mu (h) = \no{\phi}^2_{L^2} .
$$
This shows that $M\in\luno{G,\mu;\trh}_+$, and, if $\no{\phi}^2_{L^2} = 1$, then $M\in\luno{G,\mu;\trh}_1$.
\end{proof}

For each $\omega\in\sh$, recall the definition of the $W$-covariant instrument $\ii^\omega$ on $G$ given in eq.~\eqref{eq:standard-ins}:
\begin{equation*}
\ii^\omega_X (T) = \trped{2}{(\id\otimes \Ao(X)) L (T \otimes \omega) L^\ast} \quad \forall X\in\bor{G},\, T\in\trh \, .
\end{equation*}

\begin{lemma}
If $f\in\hh$ and $X\in\bor{G}$, then $\ii^\omega_X (T_f)$ is the integral operator on $\hh$ with kernel
\begin{equation}\label{eq:iihX(Pf)}
\lfq \ii^\omega_X (T_f) \rgq (x;y) = f(x) \overline{f(y)} \tr{\Ao (X) U_x \omega U_y^\ast} .
\end{equation}
\end{lemma}
\begin{proof}
By Lemma \ref{lem:1}, there exists a Hilbert space $\vv$ and a function $\phi\in\ldue{G,\lam;\vv}$ with $\no{\phi}_{L^2} = 1$ such that
$$
\omega(x;y) = \scal{\phi(x)}{\phi(y)} \quad \textrm{for } \lam\otimes\lam \textrm{ - a.~a.~} (x,y) \,.
$$
The operator $(\id\otimes \Ao(X)) L (T_f \otimes \omega) L^\ast$ is then the integral operator on $\ldue{G\times G,\lam\otimes\lam}$ with kernel
\begin{eqnarray*}
&& \lfq (\id\otimes \Ao(X)) L (T_f \otimes \omega) L^\ast \rgq (x_1,x_2 ; y_1,y_2)   = \\
&& \qquad \qquad \qquad \qquad = \scal{f(x_1) 1_X (x_2) \phi(x_2 - x_1)}{f(y_1) \phi(y_2 - y_1)} \, .
\end{eqnarray*}
By Lemmas \ref{lem:1} and \ref{lem:1.5}
\begin{eqnarray*}
\lfq \ii^\omega_X (T_f) \rgq (x;y) & = & f(x)  \overline{f(y)} \int \scal{1_X (z) \phi(z - x)}{\phi(z - y)} \de\lam (z) \\
& = & f(x) \overline{f(y)} \tr{\Ao (X) U_x \omega U_y^\ast} \, ,
\end{eqnarray*}
since $\Ao (X) U_x \omega U_y^\ast$ is the integral operator on $\hh$ with kernel
$$
[\Ao (X) U_x \omega U_y^\ast] (z;t) = \scal{1_X (z) \phi(z - x)}{\phi(t - y)} \, . 
$$
\end{proof}

If $\Mo\in\mm(G;\trh)$, with $\de\Mo (x) = M(x) \de\mu (x)$, its Fourier transform is
$$
\ff \Mo (\gamma) = \int \overline{\gamma (x)} M(x) \de\mu (x) \quad \forall \gamma\in\hG ,
$$
where the integral is defined in the sense of Bochner. Note that $\ff \Mo$ is a continuous map from $\hG$ into $\trh$.

\begin{lemma}\label{lem:inj}
Let $\Mo_1 , \Mo_2 \in\mm(G;\trh)$. If
\begin{equation}\label{eq:inj}
\tr{V_\chi U_x \ff \Mo_1 (\gamma)} = \tr{V_\chi U_x \ff \Mo_2 (\gamma)} \quad \forall x\in G,\, \forall\chi,\gamma\in\hG ,
\end{equation}
then $\Mo_1 = \Mo_2$.
\end{lemma}

\begin{proof}
We first prove the following reconstruction formula\footnote{Eq.~\eqref{eq:reconstr} follows directly from square integrability of the Schr\" odinger representation $W$ of the Weyl-Heisenberg group $\WH$, see \cite[\S 5]{Borel72} for the definition of square integrable representations and for more details on this topic.} for elements $T\in\trh$:
\begin{equation}\label{eq:reconstr}
\int \tr{V_\chi U_x T} \scal{U_x^\ast V_\chi^\ast f_1}{f_2} \de (\lam \otimes \hlam) (x,\chi) = \frac{1}{c^2} \scal{Tf_1}{f_2} \quad \forall f_1 , f_2 \in\hh .
\end{equation}
To do this, choose $\phi_1, \phi_2 \in \ldue{G,\lam;\vv}$ such that $T(x;y) = \scal{\phi_1 (x)}{\phi_2 (y)}$. Then
$$
(V_\chi U_x T)(z;t) = \scal{\chi (z) \phi_1 (z-x)}{\phi_2 (t)} ,
$$
and Lemma \ref{lem:1} yelds
$$
\tr{V_\chi U_x T} = \int \chi (z) \scal{\phi_1 (z-x)}{\phi_2 (z)} \de \lam(z) = \frac{1}{c} [(\ff \otimes \id) \Phi_{\phi_1, \phi_2}] (\chi^{-1} , x) ,
$$
where $\Phi_{\phi_1, \phi_2} \in \ldue{G\times G , \lam\otimes \lam} = \ldue{G , \lam} \otimes \ldue{G , \lam}$ is the function $\Phi_{\phi_1, \phi_2} (z,x) = \scal{\phi_1 (z-x)}{\phi_2 (z)}$. In a similar way, one obtains
$$
\scal{V_\chi U_x f_2}{f_1} = \frac{1}{c} [(\ff \otimes \id) \Phi_{f_2, f_1}] (\chi^{-1} , x) ,
$$
with $\Phi_{f_2, f_1} \in \ldue{G\times G , \lam\otimes \lam}$ given by $\Phi_{f_2,f_1} (z,x) = f_2 (z-x) \overline{f_1 (z)}$. By unitarity of Fourier transform we then have
\begin{eqnarray*}
&&\int \tr{V_\chi U_x T} \scal{U_x^\ast V_\chi^\ast f_1}{f_2} \de (\lam \otimes \hlam) (x,\chi) \\
&& \qquad \qquad = \frac{1}{c^2} \int [(\ff \otimes \id) \Phi_{\phi_1, \phi_2}] (\chi^{-1} , x) \overline{[(\ff \otimes \id) \Phi_{f_2, f_1}] (\chi^{-1} , x)} \de (\lam \otimes \hlam) (x,\chi) \\
&& \qquad \qquad = \frac{1}{c^2} \int \Phi_{\phi_1, \phi_2} (z , x) \overline{\Phi_{f_2, f_1} (z , x)} \de (\lam \otimes \lam) (x,z) \\
&& \qquad \qquad = \frac{1}{c^2} \int \scal{\phi_1 (x)}{\phi_2 (z)} \overline{f_2 (x)} f_1 (z) \de (\lam \otimes \lam) (x,z) \\
&& \qquad \qquad = \frac{1}{c^2} \scal{Tf_1}{f_2} ,
\end{eqnarray*}
which is formula \eqref{eq:reconstr}.

Now, if eq.~\eqref{eq:inj} holds for $\Mo_1 , \Mo_2 \in\mm(G;\trh)$, then, replacing $T$ with $\ff \Mo_i (\gamma)$ in reconstruction formula \eqref{eq:reconstr}, we see that $\scal{\ff \Mo_1 (\gamma) f_1}{f_2} = \scal{\ff \Mo_2 (\gamma) f_1}{f_2}$ for all $f_1,f_2\in\ldue{G,\lam}$ and $\gamma\in\hG$. Thus, $\ff \Mo_1 = \ff \Mo_2$, and $\Mo_1 = \Mo_2$ follows by injectivity of Fourier transform (see e.~g.~\cite[Theorem 4.33]{CAHA95}).
\end{proof}

After all this preparation, we are now in position to prove our main Theorem \ref{th:gen}.

\begin{proof}[Proof of Theorem \ref{th:gen}]
We divide the proof into several steps.

1) If $\Mo\in\mm (G;\trh)_1$, with $\de\Mo(x) = M(x) \de \mu (x)$, then it is easily checked that the expression under the integral in eq.~\eqref{eq:gen} is in $\luno{G,\mu;\ti{\hh\otimes\hh}}$. Looking at the equivalent formula \eqref{eq:genterza}, one immediately concludes that $\ii$ defined in eq.~\eqref{eq:gen} is a $W$-covariant instrument on $G$, since each $\ii^{\omega (x)}$ is.

2) Conversely, suppose $\ii$ is a $W$-covariant instrument on $G$. By Theorem 1 in \cite{CaHeTo09} this is equivalent to assume that there exist\footnote{We refer to Chapter VI of \cite{GQT85} for more details on systems of imprimitivity, induced representations and the Imprimitivity Theorem.}
\begin{itemize}
\item a transitive system of imprimitivity $(D,\Ro)$ for $\WH
$ based on $G\simeq \WH/N$ and acting in a Hilbert space $\kk$
\item an isometry $L:\hh \frecc \hh\otimes\kk$ intertwining the representations $W$ and $W\otimes D$
\end{itemize}
such that
\begin{equation}\label{eq:struttura:I}
\ii_X (T) = \trped{2}{\lft \id_\hh \otimes \Ro (X) \rgt L T L^\ast} \quad \forall T\in\trh \, .
\end{equation}
Moreover, the system of imprimitivity $(D,\Ro)$ and the isometry $L$ can be chosen in such a way that the set 
\begin{equation}\label{eq:total}
\left\{ (A\otimes \Ro (X)) L v \mid X\in \bor{G},\, A\in \lh,\, v\in\hh \right\}
\end{equation}
is total in $\hh \otimes \kk$.

For $z=(0,1,u) \in Z$, we have
\begin{eqnarray*}
(A \otimes \Ro(X)) Lf & = & u (A \otimes \Ro(X)) LW(z)f \\
& = & u (A \otimes \Ro(X)) (W(z) \otimes D(z)) Lf \\
& = & u (W(z) \otimes D(z)) (A \otimes \Ro(z^{-1} X)) Lf \\
& = & (\id_\hh \otimes D(z)) (A \otimes \Ro(X)) Lf \, .
\end{eqnarray*}
By totality of the set \eqref{eq:total} in $\hh\otimes \kk$, $\id_\hh \otimes D(z) = \id_{\hh\otimes\kk}$, i.~e.~$\left. D \right|_Z = \id_\kk$. Therefore, $D$ factors to a representation $\tilde{D}$ of the abelian group $\WH
/Z \simeq G\times \hG$.

The couple $(\tilde{D}, \Ro)$ is a transitive system of imprimitivity for the group $G\times \hG$ based on $G$, where the action of $G\times \hG$ on $G\simeq (G\times \hG)/\hG$ is
$$
(x,\chi)[y] = x+y \quad \forall (x,\chi)\in G\times \hG ,\, y\in G \, .
$$
By the Imprimitivity Theorem, $(\tilde{D}, \Ro)$ is the system of imprimitivity induced by some representation $\sigma$ of the group $\hG$. Possibly enlarging the representation $\tilde{D}$ (thus dropping the requirement that the set \eqref{eq:total} is total in $\kk\otimes \hh$, but preserving eq.~\eqref{eq:struttura:I}), we can assume that the representation $\sigma$ has constant infinite multiplicity, i.~e.~there exists a positive measure $\mu$ on $\widehat{\hG} = G$ and a separable infinite dimensional Hilbert space $\vv$ such that $\sigma$ acts on the space $\ldue{G, \mu ; \vv}$ as follows
\begin{equation*}
[\sigma (\chi) \phi] (h) = \chi(h) \phi (h) \quad \forall \phi \in \ldue{G, \mu ; \vv} \, .
\end{equation*}
(see e.~g.~\cite[Theorem 7.40]{CAHA95}). 
With this assumption, the inducing construction gives
\begin{equation*}
\kk = \ldue{G , \lambda ; \ldue{G, \mu ; \vv}} = \ldue{G\times G , \lambda \otimes \mu ; \vv} \, , 
\end{equation*}
and
\begin{eqnarray*}
[\tilde{D} (t,\chi) f] (y,h) & = & \chi(h) f(y-t, h) \\
\left[ \Ro(X) f \right] (y,h) & = & 1_X (y) f(y, h)
\end{eqnarray*}
for all $f\in \kk$ (see Theorem 6.7 in \cite{GQT85}).

Collecting these facts, we obtain that
\begin{equation*}
\hh\otimes \kk = \ldue{G , \lambda} \otimes \ldue{G\times G , \lambda \otimes \mu ; \vv} = \ldue{G\times G\times G , \lambda \otimes \lambda \otimes \mu ; \vv}
\end{equation*}
with
\begin{eqnarray*}
\left[(W\otimes D) (t,\chi,u) f \right] (x,y,h) & = & \overline{u} \chi(x + h - t) f(x-t,y-t, h) \\
\left[(\id_\hh \otimes \Ro(X)) f \right] (x,y,h) & = & 1_X (y) f(x,y,h)
\end{eqnarray*}
for all $f\in \hh\otimes \kk$, $(t,\chi,u)\in \WH
$, $X\in\bor{G}$.

We define a unitary operator $S$ on $\hh\otimes \kk$ by
\begin{equation*}
Sf (x,y,h) = f (x+h, y-x, h) \, .
\end{equation*}
It is easy to check that
\begin{equation*}
(W\otimes D)(t,\chi,u) \, S = S\, (W(t,\chi,u)\otimes \id_\kk) \, ,
\end{equation*}
i.~e.~$S$ intertwines $W\otimes \id_\kk$ with $W\otimes D$. On the other hand, by irreducibility of $W$, every isometry $R : \hh \frecc \hh \otimes \kk$ intertwining $W$ with $W\otimes \id_\kk$ has the form
\begin{equation*}
Rf = f\otimes \phi =: R_\phi f \quad \forall f\in\hh
\end{equation*}
for some choice of $\phi \in \kk$ with $\no{\phi} = 1$, fixed by $R = R_\phi$. Combining these two facts, every isometry $L : \hh \frecc \hh\otimes\kk$ intertwining $W$ with $W\otimes D$ is given by $L = SR_\phi =: L_\phi$ for some choice of $\phi$ as before. Explicitely,
\begin{equation*}
(L_\phi f) (x,y,h) = f(x+h) \phi (y-x, h) \, .
\end{equation*}

We now evaluate the expression in eq.~\eqref{eq:struttura:I}. For $f, f^\prime \in \hh$, we have
\begin{eqnarray*}
\scal{\ii_X(T_f) f^\prime}{f^\prime} & = & \tr{T_{f^\prime}\, \trped{2}{\lft \id_\hh \otimes \Ro(X) \rgt L_\phi T_f L_\phi^\ast}} \\
& = & \tr{\lft T_{f^\prime} \otimes \Ro(X) \rgt L_\phi T_f L_\phi^\ast} \\
& = & \scal{\lft T_{f^\prime} \otimes \Ro(X) \rgt L_\phi f}{L_\phi f} \\
& = & \int f^\prime (x) \overline{f^\prime (z)} f (z+h) \overline{f(x+h)} \scal{1_X (y) \phi (y-z,h)}{\phi (y-x,h)} \\
&& \times \de (\lam \otimes \lam \otimes \lam \otimes \mu) (x,y,z,h) .
\end{eqnarray*}
Let $M\in\luno{G,\mu;\trh}_1$ be defined as in eq.~\eqref{K<-->phi}, with $\phi_1 = \phi_2 = \phi$. Then
$$
[\Ao(X) U_z M(h) U_x^\ast ] (y;t) = \scal{1_X (y) \phi (y-z,h)}{\phi(t-x,h)} ,
$$
hence, by eq.~\eqref{tr M(h)},
\begin{eqnarray*}
\scal{\ii_X(T_f) f^\prime}{f^\prime} & = & \int f^\prime (x) \overline{f^\prime (z)} f (z+h) \overline{f(x+h)} \, \tr{\Ao (X) U_z M(h) U_x^\ast} \\
&& \times \de (\lam \otimes \lam \otimes \mu) (x,z,h) \\
& = & \int f^\prime (x-h) \overline{f^\prime (z-h)} f (z) \overline{f(x)} \, \tr{\Ao (X) U_z U_h^\ast M(h) U_h U_x^\ast} \\
&& \times \de (\lam \otimes \lam \otimes \mu) (x,z,h) \, .
\end{eqnarray*}
Defining $\omega(h) = U_h^\ast M(h) U_h / \no{M(h)}_1$, $\de\nu (h) = \no{M(h)}_1 \de\mu (h)$, by eq.~\eqref{eq:iihX(Pf)} the last expression is
\begin{equation*}
\scal{\ii_X (T_f) \, f^\prime}{f^\prime} = \int \scal{\ii^{\omega (h)}_X (T_f) \, U_h f^\prime}{U_h f^\prime} \de \nu (h) .
\end{equation*}
This equation clearly holds replacing $T_f$ with any finite rank operator, and then eq.~\eqref{eq:genterza} (which is equivalent to eq.~\eqref{eq:gen}) follows by density.

3) We finally show that the correspondence $\Mo \mapsto \ii$ estabilished in eq.~\eqref{eq:gen} is injective from $\mm(G;\trh)_1$ into the set of $W$-covariant instruments.

Suppose that $\de\Mo (x) = M (x) \de\mu (x)$ is a vector measure in $\mm(G;\trh)_1$, and let $\ii$ be the instrument associated to $\Mo$ by eq.~\eqref{eq:gen}. For all $T\in\trh$, let $\Mo^T\in\mm(G;\trh)$ be the vector measure
$$
\Mo^T (X) = \ii_X (T) .
$$
Rewriting $\ii$ in the form of eq.~\eqref{eq:genbis} and applying the second SNAG formula in \eqref{eq:SNAG}, we have
$$
\ff\Mo^T (\gamma^{-1}) = \int U_x^\ast \lfq \trped{2}{(\id\otimes V_\gamma) L (T \otimes M^\prime (x)) L^\ast} \rgq U_x \de\mu (x)
$$
for all $\gamma\in\hG$, where $M^\prime (x) = U_x^\ast M(x) U_x$. It follows that, if $y\in G$, $\chi\in\hG$,
\begin{eqnarray*}
\tr{V_\chi U^\ast_y \ff\Mo^T (\gamma^{-1})} & = & \int \overline{\chi(x)} \, \tr{(V_\chi U^\ast_y \otimes V_\gamma) L (T \otimes M^\prime (x)) L^\ast} \de\mu (x) \\
& = & \tr{(V_{\chi\gamma} U^\ast_y \otimes V_\gamma U_y) (T \otimes \ff \Mo^\prime (\chi))} \\
& = & \tr{V_{\chi\gamma} U^\ast_y T} \tr{V_\gamma U_y \ff \Mo^\prime (\chi)} \, ,
\end{eqnarray*}
where $\de\Mo^\prime (x) = M^\prime (x) \de\mu (x)$. Choose $T = U_y V^\ast_{\chi\gamma} \varrho$, with $\varrho \in\sh$. Then we have
$$
\tr{V_\chi U^\ast_y \ff\Mo^T (\gamma^{-1})} = \tr{V_\gamma U_y \ff \Mo^\prime (\chi)} .
$$
Therefore, $\ii$ determines the continuous mapping
$$
(\gamma,y,\chi) \mapsto \tr{V_\gamma U_y \ff \Mo^\prime (\chi)} ,
$$
hence it determines the vector measure $\Mo^\prime$, or, equivalently, $\Mo$ by Lemma \ref{lem:inj}.
\end{proof}

We now prove the consequences of Theorem \ref{th:gen} stated in Propositions \ref{prop:dens} and \ref{prop:di-S}.

\begin{proof}[Proof of Proposition \ref{prop:dens}]
For all $f\in\hh$ we have
\begin{eqnarray*}
\scal{\tilde{\Ao} (X) f}{f} & = & \tr{\tilde{\Ao} (X) T_f} = \tr{\ii_X (T_f)} \\
& = & \int \tr{U_x^\ast \left[ \ii_X^{\omega(x)} (T_f) \right] U_x} \de \nu(x) \\
& = & \tr{\ii^\omega_X (T_f)} ,
\end{eqnarray*}
where we used eq.~\eqref{eq:genterza} and set $\omega= \int \omega(x) \de \nu(x) = \Mo^\prime (G)$. 
By eq.~\eqref{eq:iihX(Pf)},
\begin{equation}\label{eq:*}
\lfq \ii^\omega_X (T_f) \rgq (x;y) = \scal{f(x) \Ao(X) U_x \omega^{1/2}}{f(y) U_y \omega^{1/2}}_{HS} ,
\end{equation}
where $\scal{A}{B}_{HS} = \tr{AB^\ast}$ is the scalar product in the Hilbert space of Hilbert-Schmidt operators on $\hh$. Using then Lemma \ref{lem:1} to evaluate the trace,
\begin{eqnarray*}
\scal{\tilde{\Ao} (X) f}{f}
& = & \int \scal{f(x) \Ao(X) U_x \omega^{1/2}}{f(x) U_x \omega^{1/2}}_{HS} \de\lam (x) \\
& = & \int |f(x)|^2 \tr{\Ao(X-x) \omega} \de\lam (x) ,
\end{eqnarray*}
from which eq.~\eqref{eq:dens-A} follows by comparison with the definition \eqref{eq:Afuzzy} of $\Ao_\sigma$.

For the observable $\tilde{\Bo}$ we have
\begin{eqnarray*}
\scal{\tilde{\Bo} (Y) f}{f} & = & \tr{\tilde{\Bo} (Y) T_f} = \tr{\Bo(Y) \ii_G (T_f)} \\
& = & \int \tr{\Bo(Y) U_x^\ast \left[ \ii_G^{\omega(x)} (T_f) \right] U_x} \de\nu (x) \\
& = & \int \tr{\Bo(Y) \ii_G^{\omega(x)} (T_f)} \de\nu (x) \\
& = & \tr{\Bo(Y) \ii_G^\omega (T_f)} .
\end{eqnarray*}
Combining the definition $\Bo(Y) f = c\, (\ff^{-1} 1_Y) \ast f$ and eq.~\eqref{eq:*}, we have (for $\hlam(Y) < \infty$)
$$
\lfq \Bo (Y) \ii^\omega_G (T_f) \rgq (x;y) = \scal{c \int \ff^{-1} 1_Y (x-z) f(z) U_{z} \omega^{1/2} \de\lam(z)}{f(y) U_y \omega^{1/2}}_{HS} .
$$
By Lemma \ref{lem:1}
\begin{eqnarray*}
\scal{\tilde{\Bo} (Y) f}{f} 
& = & \int \scal{c \int \ff^{-1} 1_Y (x-z) f(z) U_{z} \omega^{1/2} \de\lam(z)}{f(x) U_x \omega^{1/2}}_{HS} \de\lam(x) \\
& = & c \int \ff^{-1} 1_Y (x-z) f(z) \overline{f(x)} \tr{U_{z-x} \omega} \de(\lam\otimes\lam)(x,z) .
\end{eqnarray*}
Let $\beta_\omega$ be the measure on $\hG$ given by $\beta_\omega (Y) = \tr{B(Y) \omega}$ $\forall Y\in\bor{\hG}$.
Since
$$
\tr{U_z^\ast \omega} =\int \chi(z) \de \beta_\omega (\chi) = : \ff^{-1} \beta_\omega (z) ,
$$
we have
\begin{eqnarray*}
\scal{\tilde{\Bo} (Y) f}{f}
& = & c \int \ff^{-1} 1_Y (x-z) \ff^{-1} \beta_\omega (x-z) f(z) \overline{f(x)} \de (\lam\otimes\lam)(x,z) \\
& = & c \int \ff^{-1} (1_Y \ast \beta_\omega) (x-z) f(z) \overline{f(x)} \de (\lam\otimes\lam)(x,z) ,
\end{eqnarray*}
where $1_Y \ast \beta_\omega (\chi) = \int 1_Y (\chi \gamma^{-1}) \de \beta_\omega (\gamma)$. Comparing with the definition \eqref{eq:Bfuzzy} of $\Bo_\tau$, we see that
$$
\tau_Y = 1_Y \ast \beta_\omega ,
$$
hence
$$
\tau(Y) = \tau_Y (1) = \int 1_Y (\chi^{-1}) \de\beta_\omega (\chi) = \tr{\Bo (Y^{-1}) \omega} ,
$$
which is eq.~\eqref{eq:dens-B}.
\end{proof}

\begin{proof}[Proof of Proposition \ref{prop:di-S}]
Suppose $f,f^\prime \in\hh$, with $\no{f} = 1$, and let $\omega=T_f$. Then, for $X\in\bor{G}$, $Y\in\bor{\hG}$, we have
\begin{eqnarray*}
\scal{\ii^{\omega \ast}_X (\Bo (\hY)) \, f^\prime}{f^\prime} & = & \tr{\Bo (\hY) \, \ii^{\omega}_X (T_{f^\prime})} \\
& = & \tr{(\Bo (\hY)\otimes \Ao(X)) L (T_{f^\prime} \otimes T_f) L^\ast} \\
& = & \scal{(\hat{\Ao} (\hY) \otimes \Ao(X)) (\ff \otimes \id) L (f^\prime \otimes f)}{(\ff \otimes \id) L (f^\prime \otimes f)} \, ,
\end{eqnarray*}
where $\hat{\Ao} (Y) = \ff \Bo (Y) \ff^{-1}$ is the operator on $\hhat = \ldue{\hG,\hlam}$ given by $[\hat{\Ao} (Y) \hat{f}] (\chi) = 1_Y (\chi) \hat{f} (\chi)$.
Since
\begin{eqnarray*}
[(\ff \otimes \id) L (f^\prime \otimes f)] (\chi , x) & = & c \int \overline{\chi(y)} f^\prime (y) f(x-y) \de\lam (y) \\
& = & c \scal{f^\prime}{V_\chi U_x \check{f}} \, ,
\end{eqnarray*}
the above formula rewrites
\begin{eqnarray*}
\scal{\ii^{\omega \ast}_X (\Bo (\hY)) \, f^\prime}{f^\prime} & = & c^2 \int 1_X (x) 1_{\hY} (\chi) \left| \scal{f^\prime}{V_\chi U_x \check{f}} \right|^2 \de (\lam\otimes\hlam) (x,\chi) \\
& = & \scal{\Co_{\check{\omega}} (X\times \hY) \, f^\prime}{f^\prime} \, ,
\end{eqnarray*}
with $\check{\omega}$ defined in eq.~\eqref{eq:omegacheck}. By density of finite rank operators in $\sh$, the above equation extends by continuity to all $\omega\in\sh$. If $\ii$ is given by eq.~\eqref{eq:genterza}, we then have
\begin{eqnarray*}
\scal{\ii^\ast_X (\Bo (\hY)) \, f^\prime}{f^\prime} & = & \int \scal{\ii^{\omega (x) \ast}_X (U_x \Bo (\hY) U^\ast_x) \, f^\prime}{f^\prime} \de\nu (x) \\
& = & \int \scal{\ii^{\omega (x) \ast}_X (\Bo (\hY)) \, f^\prime}{f^\prime} \de\nu (x) \\
& = & \int \scal{\Co_{\check{\omega} (x)} (X\times \hY) \, f^\prime}{f^\prime} \de\nu (x) \\
& = & \scal{\Co_{\check{\omega}} (X\times \hY) \, f^\prime}{f^\prime} \de\nu (x) \, ,
\end{eqnarray*}
with $\omega = \int \omega (x) \de\nu (x) = \int M^\prime (x) \de\mu (x) = \Mo^\prime (G)$.
\end{proof}


\begin{thebibliography}{10}

\bibitem{AlPr77a}
S.T. Ali and E.~Prugove\v{c}ki.
\newblock Classical and quantum statistical mechanics in a common {L}iouville
  space.
\newblock {\em Phys. A}, 89:501--521, 1977.

\bibitem{AnBaAs05}
E.~Andersson, S.M. Barnett, and A.~Aspect.
\newblock Joint measurements of spin, operational locality, and uncertainty.
\newblock {\em Phys. Rev. A}, 72:042104, 2005.

\bibitem{Appleby05}
D.M. Appleby.
\newblock Symmetric informationally complete-positive operator valued measures
  and the extended clifford group.
\newblock {\em J. Math. Phys.}, 46:052107, 2005.

\bibitem{Borel72}
Armand Borel.
\newblock {\em Repr\'esentations de groupes localement compacts}.
\newblock Lecture Notes in Mathematics, Vol. 276. Springer-Verlag, Berlin,
  1972.

\bibitem{Busch85}
P.~Busch.
\newblock Indeterminacy relations and simultaneous measurements in quantum
  theory.
\newblock {\em Int. J. Theor. Phys.}, 24:63--92, 1985.

\bibitem{Busch86}
P.~Busch.
\newblock Unsharp reality and joint measurements for spin observables.
\newblock {\em Phys. Rev. D}, 33:2253--2261, 1986.

\bibitem{OQP97}
P.~Busch, M.~Grabowski, and P.J. Lahti.
\newblock {\em Operational Quantum Physics}.
\newblock Springer-Verlag, Berlin, 1997.
\newblock second corrected printing.

\bibitem{BuHeLa07}
P.~Busch, T.~Heinonen, and P.~Lahti.
\newblock Heisenberg's uncertainty principle.
\newblock {\em Phys. Rep.}, 452:155--176, 2007.

\bibitem{BuHe08}
P.~Busch and T.~Heinosaari.
\newblock Approximate joint measurements of qubit observables.
\newblock {\em Quant. Inf. Comp.}, 8:0797--0818, 2008.

\bibitem{CaHeTo04}
C.~Carmeli, T.~Heinonen, and A.~Toigo.
\newblock Position and momentum observables on {$\Bbb R$} and on {${\Bbb R}\sp
  3$}.
\newblock {\em J. Math. Phys.}, 45:2526--2539, 2004.

\bibitem{CaHeTo08}
C.~Carmeli, T.~Heinonen, and A.~Toigo.
\newblock Why unsharp observables?
\newblock {\em Int. J. Theor. Phys.}, 47:81--89, 2008.

\bibitem{CaHeTo09}
C.~Carmeli, T.~Heinosaari, and A.~Toigo.
\newblock Covariant quantum instruments.
\newblock {\em Journal of Functional Analysis}, 257:3353 -- 3374, 2009.

\bibitem{CaDeLaLe00}
G.~Cassinelli, E.~{De Vito}, P.~Lahti, and A.~Levrero.
\newblock Phase space observables and isotypic spaces.
\newblock {\em J. Math. Phys.}, 41:5883--5896, 2000.

\bibitem{DaPeSa04}
G.M. D'Ariano, P.~Perinotti, and M.F. Sacchi.
\newblock Informationally complete measurements and group representation.
\newblock {\em J. Opt. B: Quantum Semiclass. Opt.}, 6:S487--S491, 2004.

\bibitem{Davies70}
E.B. Davies.
\newblock On the repeated measurements of continuous observables in quantum
  mechanics.
\newblock {\em J. Funct. Anal.}, 6:318--346, 1970.

\bibitem{QTOS76}
E.B. Davies.
\newblock {\em Quantum Theory of Open Systems}.
\newblock Academic Press, London, 1976.

\bibitem{DaLe70}
E.B. Davies and J.T. Lewis.
\newblock An operational approach to quantum probability.
\newblock {\em Comm. Math. Phys.}, 17:239--260, 1970.

\bibitem{Flammia06}
S.~Flammia.
\newblock On {SIC-POVM}s in prime dimensions.
\newblock {\em J. Phys. A: Math. Gen.}, 39:13483--13493, 2006.

\bibitem{HAPS89}
G.B. Folland.
\newblock {\em Harmonic analysis in phase space}, volume 122 of {\em Annals of
  Mathematics Studies}.
\newblock Princeton University Press, Princeton, NJ, 1989.

\bibitem{CAHA95}
G.B. Folland.
\newblock {\em A Course in Abstract Harmonic Analysis}.
\newblock CRC Press, Boca Raton, FL, 1995.

\bibitem{HeWo10}
T.~Heinosaari and M.M. Wolf.
\newblock Nondisturbing quantum measurements.
\newblock {\em J. Math. Phys.}, 51:092201, 2010.

\bibitem{Hille-Ph}
E. Hille and R.~S. Phillips.
\newblock {\em Functional analysis and semi-groups}.
\newblock American Mathematical Society Colloquium Publications, vol. 31.
  American Mathematical Society, Providence, R. I., 1957.
\newblock rev. ed.

\bibitem{Holevo73}
A.S. Holevo.
\newblock Optimal quantum measurements.
\newblock pages 1172--1177, 1973.

\bibitem{PSAQT82}
A.S. Holevo.
\newblock {\em Probabilistic and Statistical Aspects of Quantum Theory}.
\newblock North-Holland Publishing Co., Amsterdam, 1982.

\bibitem{Husimi40}
K.~Husimi.
\newblock Some formal properties of the density matrix.
\newblock {\em Proc. Phys. Math. Soc. Japan}, 22:264--314, 1940.

\bibitem{JePu07}
A.~Jen\v{c}ov{\'a} and S.~Pulmannov{\'a}.
\newblock How sharp are {PV} measures?
\newblock {\em Rep. Math. Phys.}, 59:257--266, 2007.

\bibitem{Lang}
Serge Lang.
\newblock {\em Real analysis}.
\newblock Addison-Wesley Publishing Company Advanced Book Program, Reading, MA,
  second edition, 1983.

\bibitem{MackeyStone}
George~W. Mackey.
\newblock A theorem of {S}tone and von {N}eumann.
\newblock {\em Duke Math. J.}, 16:313--326, 1949.

\bibitem{QCQI00}
M.A. Nielsen and I.L. Chuang.
\newblock {\em Quantum computation and quantum information}.
\newblock Cambridge University Press, Cambridge, 2000.

\bibitem{OpBuBaDr95}
T.~Opatrn{\'y}, V.~Bu\v{z}ek, J.~Bajer, and G.~Drobn{\'y}.
\newblock Propensities in discrete phase spaces: Q function of a state in a
  finite-dimensional hilbert space.
\newblock {\em Phys. Rev. A}, 52:2419--2428, 1995.

\bibitem{Ozawa84}
M.~Ozawa.
\newblock Quantum measuring processes of continuous observables.
\newblock {\em J. Math. Phys.}, 25:79--87, 1984.

\bibitem{Ozawa86}
M.~Ozawa.
\newblock On information gain by quantum measurements of continuous
  observables.
\newblock {\em J. Math. Phys.}, 27:759--763, 1986.

\bibitem{MMMPI80}
M.~Reed and B.~Simon.
\newblock {\em Methods of Modern Mathematical Physics, Vol. I: Functional
  Analysis}.
\newblock Academic Press, London, revised and enlarged edition, 1980.

\bibitem{ReBlScCa04}
J.M. Renes, R.~Blume-Kohout, A.J. Scott, and C.M. Caves.
\newblock Symmetric informationally complete quantum measurements.
\newblock {\em J. Math. Phys.}, 45:2171--2180, 2004.

\bibitem{GQT85}
V.~S. Varadarajan.
\newblock {\em Geometry of quantum theory}.
\newblock Springer-Verlag, New York, second edition, 1985.

\bibitem{MFQM55}
J.~{von Neumann}.
\newblock {\em Mathematical Foundations of Quantum Mechanics}.
\newblock Princeton University Press, Princeton, 1955.
\newblock Translated by R.T. Beyer from \emph{Mathematische Grundlagen der
  Quantenmechanik}, Springer, Berlin, 1932.

\bibitem{ZiBu05}
M.~Ziman and V.~Bu\v{z}ek.
\newblock Realization of positive-operator-valued measures using
  measurement-assisted programmable quantum processors.
\newblock {\em Phys. Rev. A}, 72:022343, 2005.

\bibitem{cpso}
Several different proofs of this result are known; see
R.~Werner. \newblock {\em J. Math. Phys.}, 25:1404--1411, 1984;
G.~Cassinelli, E.~{De Vito}, and A.~Toigo. \newblock {\em J. Math. Phys.}, 44:4768--4775, 2003;
J.~Kiukas, P.~Lahti, and K.~Ylinen. \newblock {\em J. Math. Anal. Appl.}, 319:783--801, 2006.


\end{thebibliography}
\end{document}